\documentclass[amsmath,onecolumn, floatfix,nofootinbib, prx, aps, usenames, dvipsnames,11pt]{quantumarticle}
\pdfoutput=1

\usepackage[utf8]{inputenc}
\usepackage[english]{babel}
\usepackage[T1]{fontenc}
\usepackage{natbib}

\usepackage[usenames,dvipsnames]{xcolor}
\usepackage{graphicx}
\usepackage[ colorlinks = true, 
             linkcolor = MidnightBlue,
             urlcolor  = MidnightBlue,
             citecolor = orange,
             anchorcolor = ForestGreen,
]{hyperref} 
\usepackage{float}

\usepackage[sans]{dsfont} % contains font for identity character
\usepackage{bbm}
\usepackage[mathscr]{euscript}
\usepackage{cancel} %to strike out things
\usepackage{amsfonts}

\usepackage{tcolorbox}
\tcbuselibrary{theorems} 

\usepackage{amsmath}
\usepackage{amssymb}
\usepackage{units}
\usepackage{framed}
\usepackage{mdframed}
\usepackage{thmtools, thm-restate} %pretty  and restatable theorems
\usepackage{ucs}
\usepackage{subcaption}
\usepackage{fullpage}
\usepackage{csquotes}
\usepackage{epigraph}

\usepackage{multicol}
\usepackage{enumerate}
\usepackage{tikz}
\usetikzlibrary{arrows}
\usepackage{rotating, xcolor}
\usetikzlibrary{shapes}
\usetikzlibrary{hobby}
\usetikzlibrary{decorations.markings}

\makeatletter
\def\l@subsubsection#1#2{}
\def\l@subsection#1#2{}
\makeatother

\tcbset{colback=orange!5,colframe=orange!75!yellow, 
fonttitle= \bfseries, float=htb}

	\newcommand{\green}[1]{\textcolor{OliveGreen}{#1}}

	\newcommand{\lidia}[1]{\green{ [#1]}}

    \newcommand{\ket}[1]{\vert  #1 \rangle}
    \newcommand{\bra}[1]{\langle #1 |}
    \newcommand{\inprod}[2]{\langle #1 | #2 \rangle}
	\newcommand{\proj}[2]{\ket{#1}\bra{#2}}

	\renewcommand{\vec}[1]{\mathbf{#1}}

\newcommand{\footnoteremember}[2]{\footnote{\label{#1}#2}\newcounter{#1}\setcounter{#1}{\value{footnote}}}
\newcommand{\footnoterecall}[1]{\hyperref[#1]{\footnotemark[\value{#1}]}}

\declaretheorem[]{axiom}
\declaretheorem[]{definition}
\declaretheorem[sibling=definition]{theorem}

\newenvironment{proof}{\paragraph{Proof:}}{\hfill$\square$}
\newenvironment{remark}{\textit{Remark:}}{}

\usepackage{cleveref}

\input{Qcircuit}
\usepackage[page]{totalcount}

\usetikzlibrary{arrows,decorations.markings}

\begin{document}

\title{Multi-agent paradoxes beyond quantum theory}

\author{V. Vilasini}
\affiliation{Department of Mathematics, University of York, Heslington, York, YO10 5DD, UK}
\email{vv577@york.ac.uk}

\author{Nuriya Nurgalieva}
\affiliation{Institute for Theoretical Physics, ETH Zurich, 8093 Z\"{u}rich, Switzerland}
\email{nuriya@phys.ethz.ch}

\author{Lídia del Rio}
\affiliation{Institute for Theoretical Physics, ETH Zurich, 8093 Z\"{u}rich, Switzerland}
\email{lidia@phys.ethz.ch}

\date{}

\begin{abstract}
Which theories lead to a contradiction between simple reasoning principles and modelling observers' memories as physical systems? Frauchiger and Renner have shown that this is the case for quantum theory~\cite{Frauchiger2018}. 
Here we generalize the conditions of the Frauchiger-Renner result so that they can be applied to arbitrary physical theories, and in particular to those expressed as \emph{generalized probabilistic theories} (GPTs) \cite{Hardy01, Barrett07}. We then apply them to a particular GPT, box world, and find a deterministic contradiction in the case where agents may share a PR box \cite{Popescu1994}, which is stronger than the quantum paradox, in that it does not rely on post-selection. 
Obtaining an inconsistency for the framework of GPTs broadens the landscape of theories which are affected by the application of classical rules of reasoning to physical agents.
In addition, we model how observers' memories may evolve in box world, in a way consistent with Barrett's criteria  for allowed operations \cite{Barrett07, Gross2010}.
\end{abstract}

\maketitle

\setlength{\epigraphwidth}{4in}
\epigraph{Ordinary readers, forgive my paradoxes: one must make them when one reflects; and whatever you may say, I prefer being a man with paradoxes than a man with prejudices.}{Jean-Jacques Rousseau, \emph{Emile or On Education}}

\section{Motivation}
\label{sec:introduction}

\begin{figure}[t]
\centering
    \begin{subfigure}{0.4\textwidth}
       \includegraphics[scale=0.3]{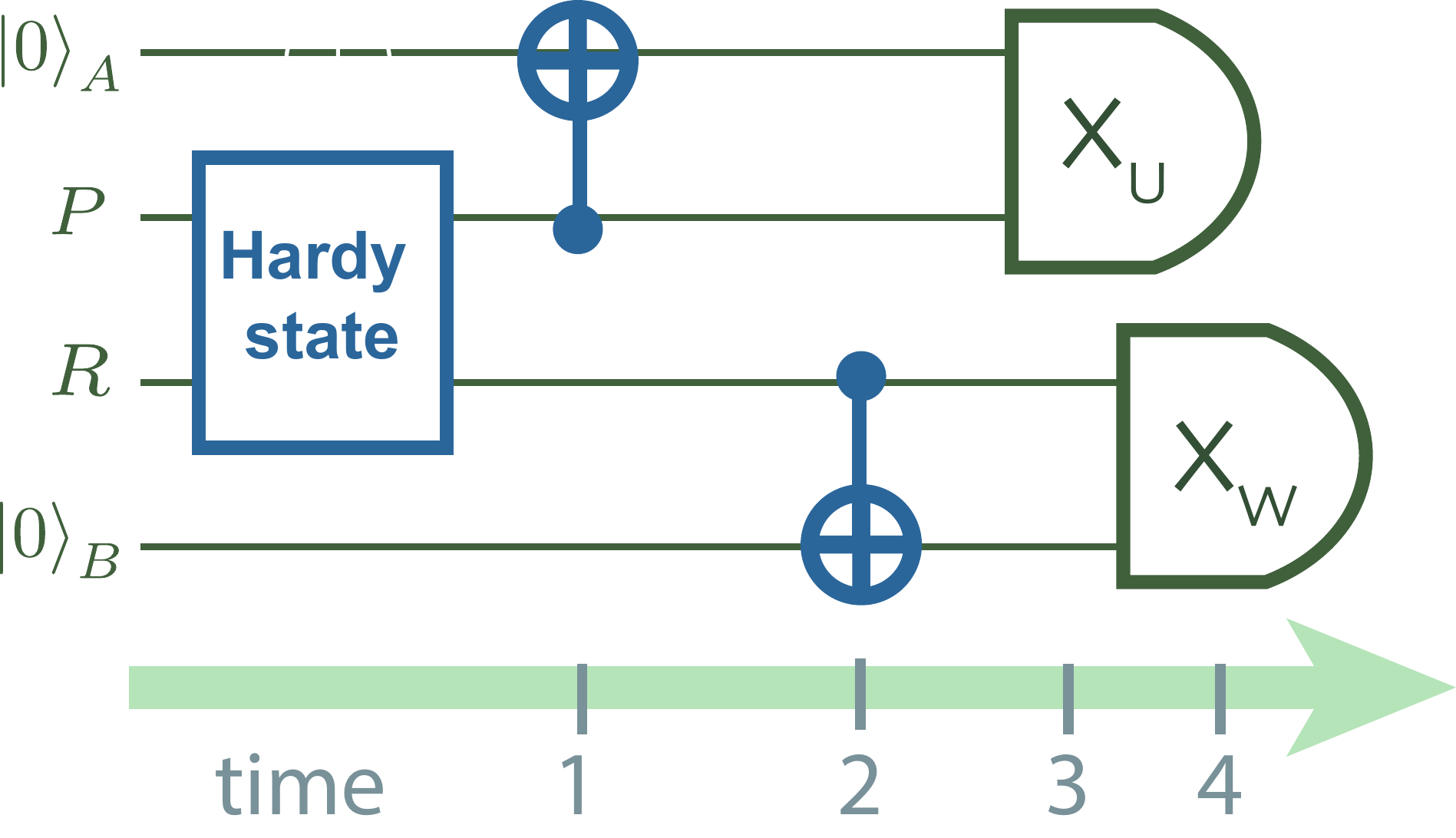}
        \caption{Circuit for the Frauchiger-Renner quantum thought experiment, where Alice and Bob share a Hardy state \cite{Frauchiger2018}. }
        \label{fig:circuitFR}
    \end{subfigure}
    \
    \begin{subfigure}{0.45\textwidth}
        \includegraphics[scale=0.3]{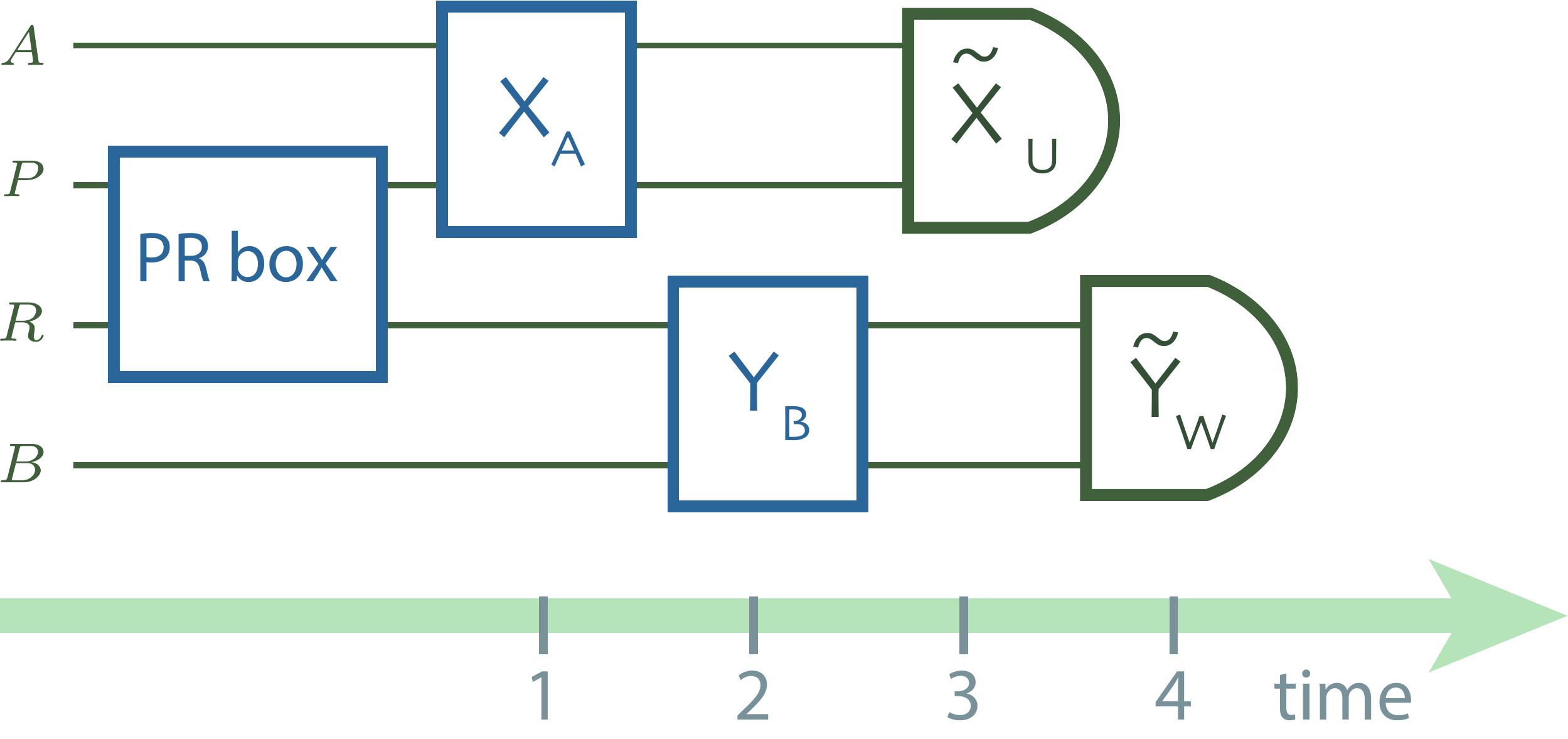}
        \caption{Timeline for the proposed thought experiment in box world, where agents Alice and Bob share a PR box.}
        \label{fig:circuitPRbox}
    \end{subfigure}
\caption{{\bf Protocols for multi-agent paradoxical experiments, as seen by outside observers.} Two instances of thought experiments which lead to a logical contradiction:  in quantum theory and in box world. {\bf a)} In the quantum case, the original experiment is formulated as a prepare-and-measure scenario  \cite{Frauchiger2018} but is equivalent to the version shown here. Alice and Bob share a Hardy state, $\ket{\Psi}_{PR}=1/\sqrt{3}(\ket{00}+\ket{10}+\ket{11})$ where the $P$ and $R$ systems correspond to Alice's and Bob's halves of the state respectively. They then measure their halves of the state in the $Z$ basis and record the outcome in their quantum memories $A$ and $B$. For outside observers Ursula and Wigner, these measurements are modelled as unitary evolutions (like the {\sc cnot} gates pictured) that correlate memory and the subsystem measured.  Finally, Ursula (and Wigner) measure $AP$ (and $RB$) in the basis $\{\ket{ok}_{AP}=1/\sqrt{2}(\ket{00}_{AP}-\ket{11}_{AP}),\ket{fail}_{AP}=1/\sqrt{2}(\ket{00}_{AP}+\ket{11}_{AP})\}$ (and analogously defined $\{\ket{ok}_{RB},\ket{fail}_{RB}\}$). 
After all is done, agents reason about each others' knowledge to find a contradiction. In the quantum case, this only happens if Ursula and Wigner obtain the outcomes $u=w=ok$ (giving the paradoxical chain $u=w=ok\Rightarrow b=1 \Rightarrow a=1 \Rightarrow w=fail$, where $a$ and $b$ are Alice's and Bob's outcomes). 
{\bf b)} In box world, Alice and Bob share a PR box and measure their halves using some measurements labelled $X=Y$ respectively (not to be confused with $X$,$Y$ basis measurements), updating their memories $A$ and $B$. Ursula and Wigner can now measure the joint systems $AP$ and $RB$ using different settings ($\tilde X=\tilde Y=1 = X \oplus 1$).
In this case, the contradiction can always be found independently of their outcomes (Section~\ref{sec:paradox}).}
\label{fig:circuits}
\end{figure}

In order to process information and make logical inferences, we would like to be able to apply simple reasoning principles to all situations. By this we mean that ideally we would like inferences such as ``if I know that $a$ holds, and I know that $a$ implies $b$, then I know that $b$ holds'' to be valid independently of the nature of $a$ and $b$ --- to take logic as a primitive that can be applied to any physical setting. When considering scenarios with several rational agents,  this extends to reasoning about each other's knowledge. Examples include games like poker, complex auctions, cryptographic scenarios, and of course  \href{https://en.wikipedia.org/wiki/Hat_puzzle}{logical hat puzzles}, where we must process complex statements of the sort ``I know that she knows that he does not know $a$'' to keep track of the flows of knowledge. 

On the other hand, when we describe the world through physics, we would like to consider ourselves a part of it, and in particular we would like to model our brains and memories as physical systems described by some theory. 
When that theory is quantum mechanics, it turns out that these two desiderata (applying to reason about each other's knowledge, and modelling memories as physical systems) are incompatible. This was first pointed out by Frauchiger and Renner, in a thought experiment where agents who can measure each others  memories (modelled as quantum systems) and reason about shared and individual knowledge may reach contradictory conclusions \cite{Frauchiger2018}. We will not review the original experiment here, apart from a very brief description in Figure~\ref{fig:circuitFR}; a pedagogical exposition can be found in our paper~\cite{NL2018}, but is not necessary to follow this article.

Our ultimate goal is  to understand whether this incompatibility between multi-agent logic and physics is a peculiar feature of quantum theory, or if  other physical theories also admit this kind of contradictions.
If the latter is true, we would like to  outline a class of theories where these logical inconsistencies may arise. Such an analysis  could help us identify the features of quantum theory responsible for such a paradox; in particular, here we investigate the landscape of generalized probabilistic theories \cite{Hardy01, Barrett07}.

\paragraph{Contributions of this work.} In Section~\ref{sec:conditions}, we generalize conditions on reasoning, memories and measurements so that they can be applied to any physical theory. The conditions can be briefly summarized as: 
 agents may use logic to reason about each others' knowledge; a physical theory allows agents to make predictions about the outcomes of measurements; and a measurement by an agent Alice  may be modelled by others as a physical evolution on her lab which preserve the information about the original system measured (from the outside agents' perspective). This generalizes the von Neumann view of measurements as a unitary evolution of the system and measurement apparatus \cite{vonNeumann1955}.
In Section~\ref{sec:boxworld} we apply those conditions to the framework of \textit{generalized probabilistic theories} (GPTs) \cite{Hardy01, Barrett07}; in particular we introduce a way to describe an agent's measurement from the perspective of other agents in the particular GPT of box world. 
Finally, in Section~\ref{sec:paradox} we derive a logical inconsistency akin to one found in \cite{Frauchiger2018}, using a setup where agents share a PR box, a maximally non-local resource in box world. The paradox found is stronger than the quantum one, in the sense that it does not rely on post-selection: agents always reach a contradiction, independently of the outcome\footnote{The joint state and the probability distributions of the original Frauchiger-Renner paradox are akin to those of Hardy's paradox~\cite{Hardy1993}. For a comparison of Hardy's paradox and PR box and why the latter allows for a contradiction without post-selection, see \cite{Abramsky15}.}. 
A high-level circuit representation of the original experiment, as well as  the PR box version, are depicted in Figure \ref{fig:circuits}.

\section{Generalized reasoning, memories and measurements}
\label{sec:conditions}
Here we generalize the Frauchiger-Renner conditions for inter-agent consistency to general physical theories. The  conditions can be instantiated by each specific theory.  This includes but is not limited to theories framed in the approach of generalized probabilistic theories \cite{Hardy01}. 
In some theories, like quantum mechanics and box world (a GPT), we will find these four conditions to be incompatible, by finding a direct contradiction in examples like the Frauchiger-Renner experiment or the PR-box experiment described in Section~\ref{sec:paradox}. In other theories (like classical mechanics and Spekkens' toy theory \cite{Spekkens07}) these four conditions may be compatible.  A complete characterization of theories where one can find these paradoxes is the subject of future work.

\subsection{Reasoning about knowledge}

\begin{figure}[t]
\centering
\begin{subfigure}{0.45\textwidth}
    \vspace{1cm}
    \includegraphics[scale=0.105]{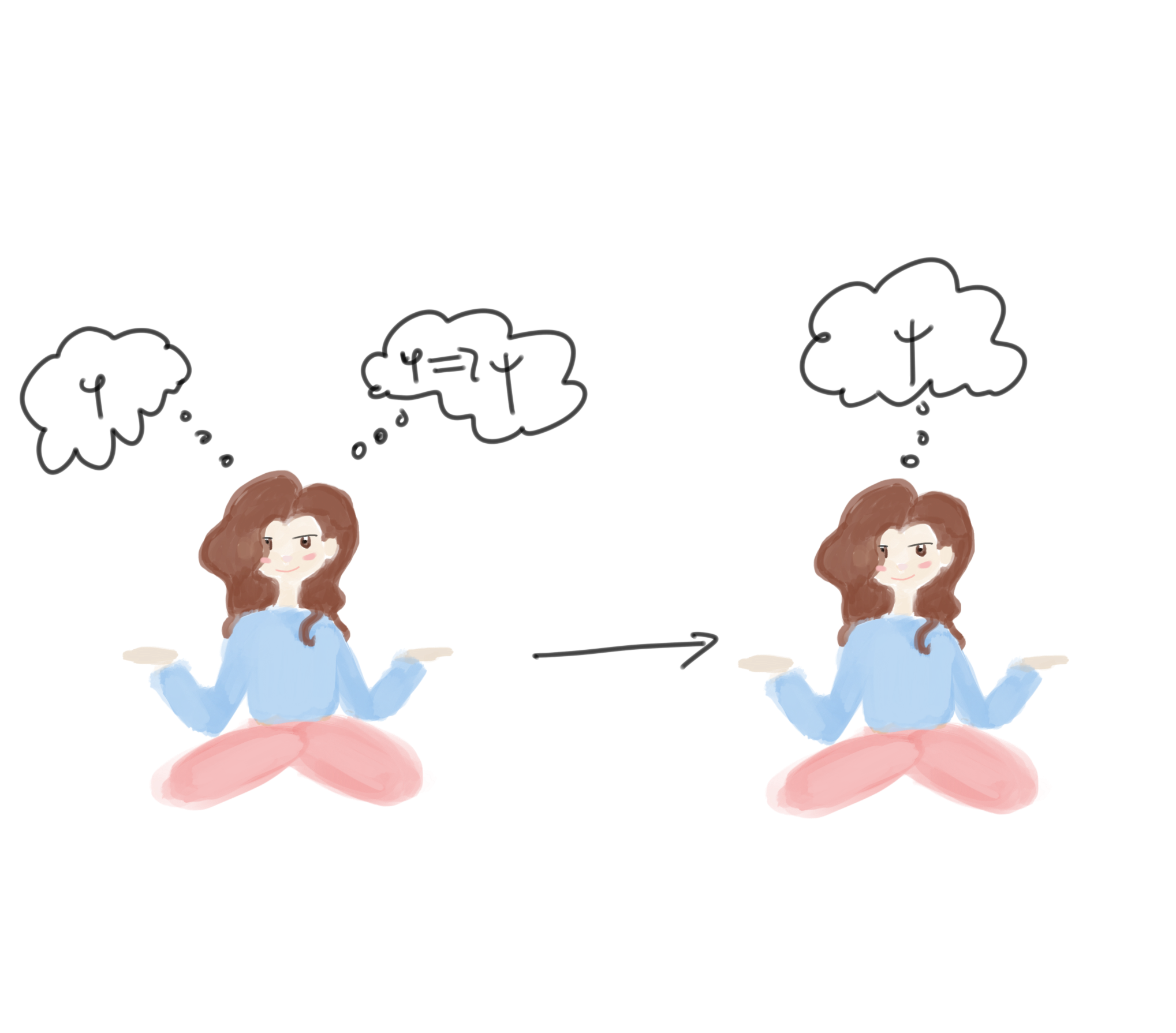}
    \caption{An agent using deduction, applying the distribution axiom of modal logic.}
    \label{fig:distribution}
\end{subfigure}
\
\begin{subfigure}{0.45\textwidth}
    \includegraphics[scale=0.12]{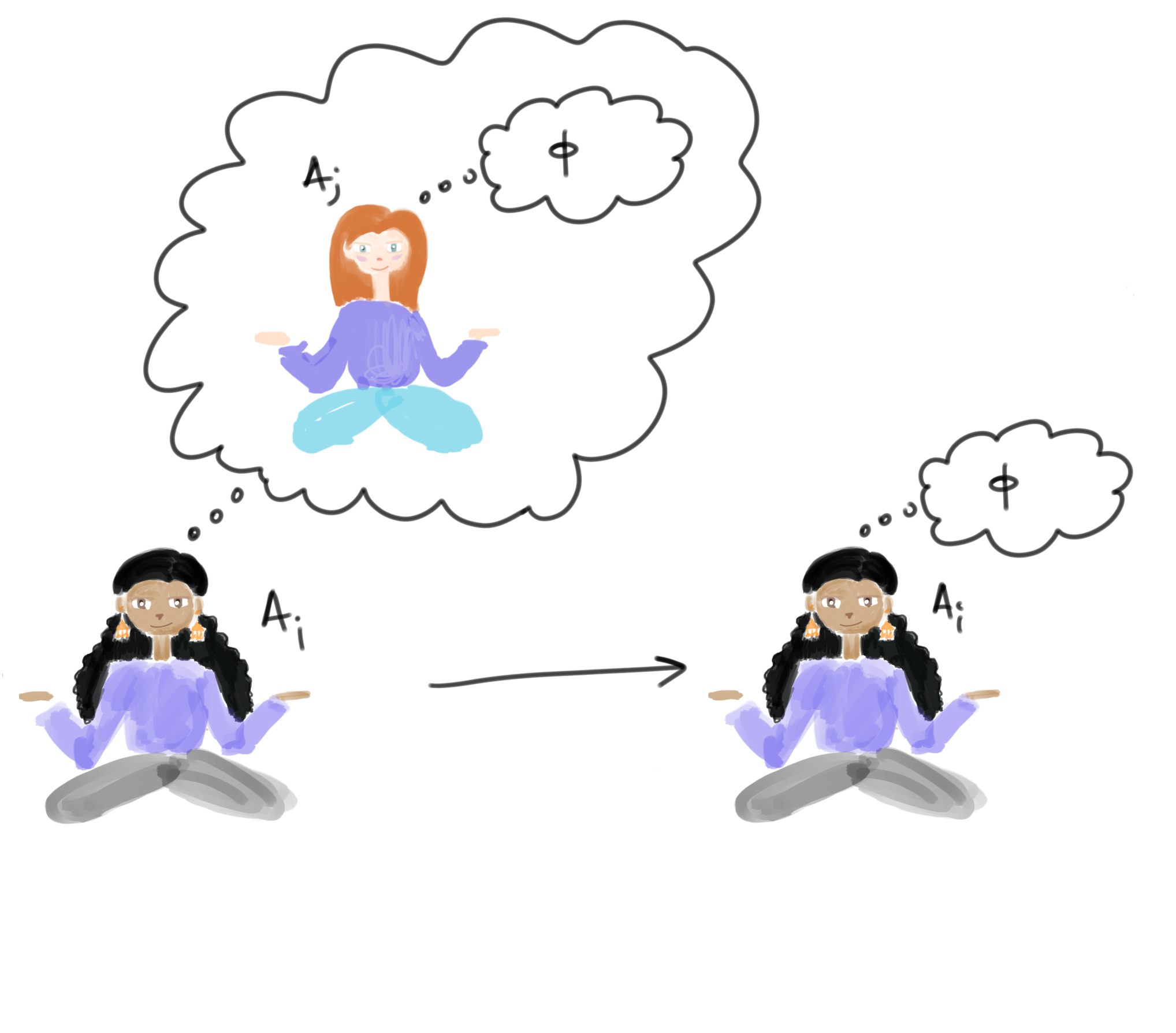}
     \caption{An agent $A_i$ trusts another agent $A_j$, denoted by $A_j \leadsto A_j$, if they take all of $A_j$'s knowledge to be true for $A_i$ as well.}
    \label{fig:trust}
\end{subfigure}
\caption{\textbf{Agents use logic to reason.} A desiderata for useful physical theories is that agents be allowed to make deductions and transfer knowledge from one another, given a trust relation (Definition \ref{def:reasoning}). For a short review of the modal logic framework and axioms, see Appendix \ref{appendix:logic}.}
\end{figure}

This condition is theory-independent. It tells us that rational agents can reason about each other's knowledge in the usual way. This is formalized by a weaker version of \emph{epistemic modal logic}, which we explain in the following (for the full derivation of the form used here see \cite{NL2018}).

Let us start with a simple example.
The goal of modal logic is to allow us to operate with chained statements like ``Alice knows that Bob knows that Eve doesn't know the secret key $k$, and Alice further knows that $k=1$,''  which can be expressed as $$K_A\ [(K_B\ \neg K_E\  k ) \ \wedge \ k=1],$$ where the operators $K_i$ stand for ``agent $i$ knows.''
If in addition Alice trusts Bob to be a rational, reliable agent, she can deduce from the statement ``I know that Bob knows that Eve doesn't know the key'' that ``I know that Eve doesn't know the key'', and forget about the source of information  (Bob). This is expressed as   $$K_A ( K_B\  \neg K_E \ k) \implies K_A \ \neg K_E \ \ k.$$
We should also allow Alice to make deductions of the type ``since Eve does not know the secret key, and one would  need to know the key in order to recover the encrypted message $m$, I conclude that Eve cannot know the secret message,'' which can be encoded as $$K_A [(\neg K_E\  k) \wedge (K_i\ m \implies K_i\ k,\ \forall \ i)] \implies K_A \neg K_E\  m.$$ 
Generalizing from this example, this gives us the following structure. 

\begin{definition}[Reasoning agents]
\label{def:reasoning}
An experimental setup with multiple agents $A_1, \dots A_N$ can be described by knowledge operators $K_1, \dots K_N$ and statements $\phi \in \Phi$, such that $K_i \phi$ denotes ``agent $A_i$ knows $\phi$.'' It should allow agents to make \emph{deductions} (Figure \ref{fig:distribution}), that is
 $$K_i [\phi \wedge (\phi \implies \psi) ] \implies K_i\ \psi.$$
 
Furthermore, each experimental setup defines a \emph{trust relation} between agents (Figure \ref{fig:trust}): we say that an agent $A_i$ trusts another agent $A_j$ (and denote it by $A_j \leadsto A_j$) iff  for all statements $\phi$, we have $$ K_i (K_j \ \phi) \implies K_i\ \phi.$$
\end{definition}

For the purposes of following the example of Section~\ref{sec:paradox}, this informal definition suffices. 
The full formal version of the axioms of modal logic used here can be found in Appendix~\ref{appendix:logic}.\footnote{Note that in general `one human $\neq$ one agent.' For example, consider a setting where we know that Alice's memory will be tampered with at time $\tau$ (much like the original Frauchiger-Renner experiment, or the sleeping beauty paradox \cite{Elga2000}). We can define  two different agents $A_{t<\tau}$ and $A_{t>\tau}$ to represent Alice before and after the tampering --- and then for example Bob could trust pre-tampering (but not post-tampering) Alice, $A_{t<\tau} \leadsto B$.} 

\paragraph{A note on the complexity cost of reasoning.}
Note that in general, even the most rational physical agents may be limited by bounded processing power and memory, will not be able to chain an indefinite number of deductions within sensible time scales. That is, these axioms for reasoning are an idealization of absolutely rational agents with unbounded processing power (see~\cite{Aaronson2017} for an overview of this and related issues). If we would like modal logic to apply to realistic, physical agents, we might account for a cost (in time, or in memory) of each logical deduction, and require it to stay below a given threshold, much like a resource theory for complexity. However, in the examples of this paper, agents only need to make a handful of logical deductions, and these complexity concerns do not play a significant role.

\subsection{Physical theories as common knowledge}

This condition is to be instantiated by each physical theory, and is the way that we incorporate the physical theory into the reasoning framework used by agents in a given setting.
If all agents use the same  theory to model the operational experiment (like quantum mechanics, special relativity, classical statistical physics, or box world), this is included in the \emph{common knowledge} shared by the agents.  
For example, in the case of quantum theory, we have that ``everyone knows that the probability of obtaining outcome $\ket x$ when measuring a state $\ket\psi$ is given by $|\inprod x\psi|^2$, and everyone knows that everyone knows this, and so on.''

\begin{figure}
    \centering
    \includegraphics[scale=0.12]{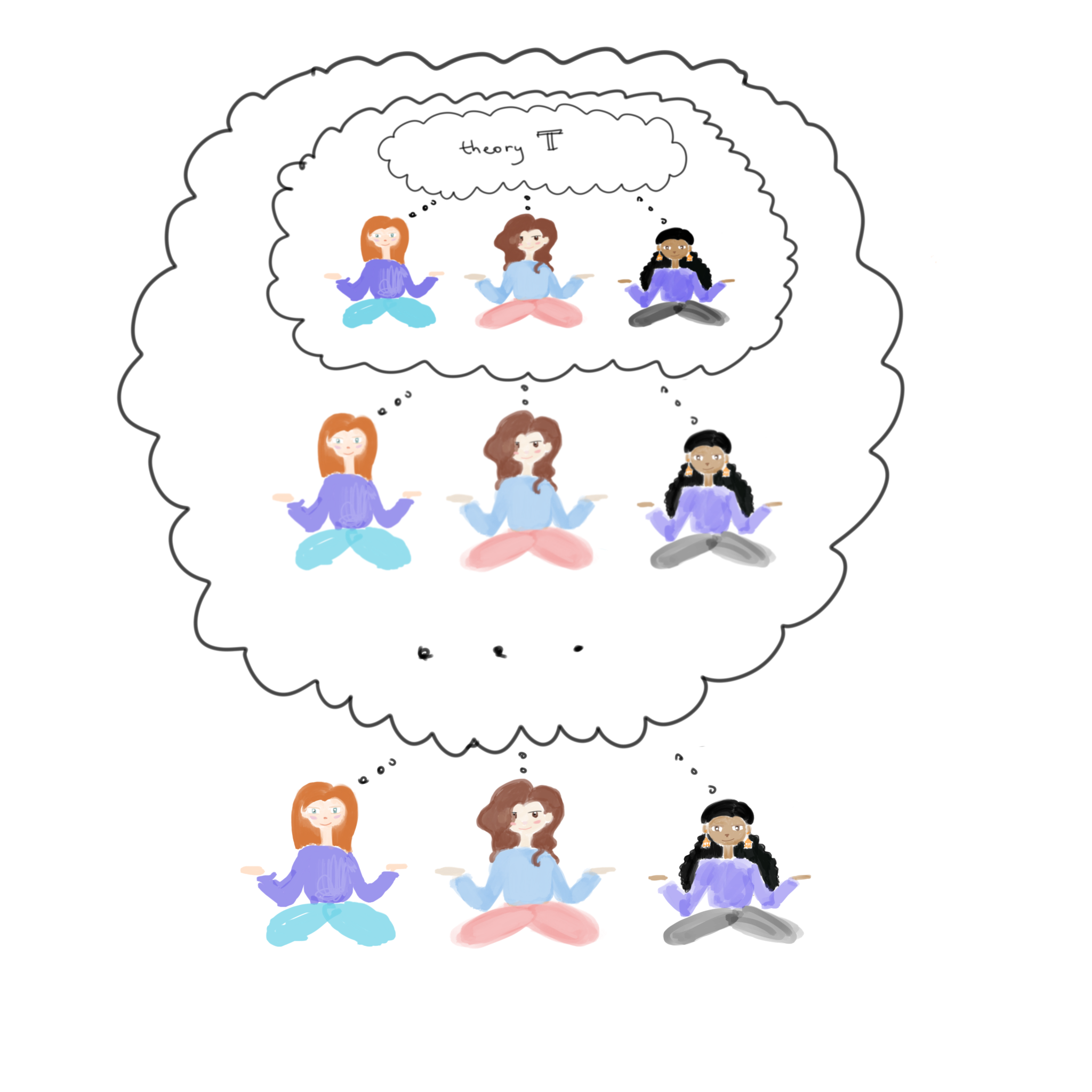}
    \caption{\textbf{Common knowledge.} Here, a shared physical theory $\mathbb T$ is common knowledge: all agents know that all agents know that ... (and so on) ... that theory $\mathbb T$ holds.}
    \label{fig:theory}
\end{figure}

\begin{definition}[Common knowledge]
We model a physical theory shared by all agents $\{A_i\}_i$ in a given setting as a set $\mathbbm T$ of statements that are common knowledge shared by all agents, i.e.
$$ \phi \in \mathbb T \iff  (\{K_i\}_i)^n \ \phi, \quad \forall\ n\in \mathbb N,$$
where $(\{K_i\}_i)^n$ is the set of all possible sequences of $n$ operators picked from $\{K_i\}_i$. For example, $(K_1 \ K_5 \ K_1 \ K_2) \in (\{K_i\}_i)^4$ and stands for ``agent $A_1$ knows that agent $A_5$ knows that agent $A_1$ knows that agent $A_2$ knows.'' 
\end{definition} 

Note that the set $\mathbb T$ of common knowledge may include statements about the settings of the experiment, as well as complex derivations \footnote{One can also alternatively model a physical theory as a subset $\mathbb T_P$ of the set $\mathbb T$ of common knowledge, $\mathbb T_P\subseteq\mathbb T$, in the case when details of experimental setup are not relevant to the theoretical formalism.}. 
To find our paradoxical contradiction, we may only need a very weak version of a full physical theory: for example Frauchiger and Renner only require a possibilistic version of the Born rule, which tells us whether an outcome will be observed with certainty \cite{Frauchiger2018}. This will also be the case in box world.

\subsection{Agents as physical systems}

In operational experiments, a reasoning agent can make statements about systems that she studies; consequently, the theory used by the agent must be able to produce a description or a model of such a system, namely, in terms of a set of states. For example, in quantum theory a two-state quantum system with a ground state $\ket{0}$ and an excited state $\ket{1}$ (\textit{qubit}) can be fully described by a set of states $\{\ket{\psi}\}$ in a Hilbert space $\mathcal{H}$, where $\ket{\psi}=\alpha\ket{0}+\beta\ket{1}$ with $\alpha,\beta\in\mathbb{C}$ and $|\alpha|^2+|\beta|^2=1$. Another examples of theories and respective descriptions of states of systems include: GPTs, where e.g.~a generalised bit (\textit{gbit}) is a system completely characterized by two binary measurements which can be performed on it \cite{Barrett07} (a review of GPTs can be found in Section \ref{sec:boxworld}); algebraic quantum mechanics, with states defined as linear functionals $\rho:A\to\mathbb{C}$, where $A$ is a $C*$-algebra \cite{vonNeumann1955}; or resource theories with some state space $\Omega$, and epistemically defined subsystems \cite{delRio2015, Kraemer2018}.

\begin{definition}[Systems]
\label{def:systems}
Here we call a ``physical system'' (or simply ``system'') anything that can be an object of a physical study\footnote{We strive to be as general as possible and do not suppose or impose any structure on systems and connections between them; in particular, we don't make any assumptions about how composite systems are formally described in terms of their parts.}. A system can be characterized, according to the theory $\mathbb{T}$, by a set of states $\{P_S^i\}_{i\in\mathcal{I}_S}$ ($\mathcal{I}_S \subseteq\mathbb{N}$).
\end{definition}

We have already used knowledge operators $K_i$ to denote knowledge of each agent. Now let us add memory to the formal description of an agent.
\begin{definition}[Agents]
\label{def:agents}
A physical setting may be associated with a set $\mathcal A$ of agents. 
An agent $A_i \in\mathcal{A}$ is described by a knowledge operator $K_i\in\mathcal{K_\mathcal{A}}$ and a physical system $M_i\in\mathcal{M_\mathcal{A}}$, which we call a ``memory.'' Each agent may study other systems according to the theory $\mathbb T$. An agent's memory $M_i$ records the results and the consequences of the studies conducted by $A_i$. The memory may be itself an object of a study by other agents.
\end{definition}

\subsection{Measurements and memory update}
Here we consider measurements both from the perspective of an agent who performs them, and that of another agent who is modeling the first agent's memory.

In an experiment involving measurements, each agent has the subjective experience of only observing one outcome (independently of how others may model her memory), and we can see this as the definition of a measurement: if there is no subjective experience of observing a single outcome, we don't call it a measurement. We can express this experience as statements such as $\phi_0 =$ ``The outcome was 0, and the system is now in state $\ket0$.'' Let us explain further after the formal definition. 

\begin{definition}[Measurements]
\label{def:measurements}
A measurement is a type of study that can be conducted by an agent $A_i$, while studying a system $S$; the essential result of the study is the obtained ``outcome'' $x\in\mathcal{X}_S$. If \emph{witnessed} by another agent $A_j$ (who knows that $A_i$ performed the measurement but does not know the outcome), the measurement is characterized by a set of propositions $\{\phi_x\}\in\Phi$, where $\phi_x$ corresponds to the outcome $x$, satisfying:
\begin{itemize}
\item $K_j(\exists \ x\in\mathcal{X}_S: K_i\ \phi_x)$,
\item $K_j \ K_i\ \phi_x \implies K_j \ K_i\ \neg (\phi_y), \quad \forall \ y \neq x$.
\end{itemize}
\end{definition}
The first condition tells us that $A_j$ knows that $A_i$ must have observed one outcome, and derived all the relevant conclusions, as expressed by one of the propositions $\phi_x$. For example, if the measurement represents a perfect $Z$ measurement of a qubit, $\phi_0$ may include statements like ``the qubit is now in state $\ket0$; before the measurement it was not in state $\ket1$; if I measure it again in the same way, I will obtain outcome 0;'' and so on.
The second condition roughly implements the experience of observing a single outcome and trusting that information. If $A_i$ observes $x$, they conclude that the conclusions $\phi_y$ that they would have derived had they observed a different outcome $y$ are not valid. In the previous example, they would know that it does not hold $\phi_1=$ ``the qubit is now in state $\ket1$; before the measurement it was not in state $\ket0$; if I measure it again I will see outcome 1.'' This condition also ensures that the conclusions $\{\phi_x\}_x$ are mutually incompatible, i.e.\ that the measurement is tightly characterized.

A measurement of another agent's memory is also an example of a valid measurement. In other words, agent $A_j$ can choose  $A_i$'s lab, consisting of $A_i$'s memory and another system $S$ (which $A_i$ studies), as an object of her study.

Thus, any agent's memory can be modelled by the other agents as a physical system undergoing an evolution that correlates it with  the measured system. In quantum theory, this corresponds to the unitary evolution 
\begin{align}
    \left( \sum_{x=0}^{N-1} p_x\ \ket x_{\text{system}}\right) \otimes \ket{0}_{\text{memory}} \to  \sum_{x=0}^{N-1} p_x \underbrace{\ket x_{\text{system}} \otimes \ket{x}_{\text{memory}}}_{=: \ \ket{\tilde x}_{SM}}.
    \label{eq:entangling_measurement}
\end{align}
The key aspect here is that the set of states of the joint system of observed system and memory, 
$\{P_{SM}^l\}_l = \operatorname{span} \{  \ket x_{\text{system}} \otimes \ket{x}_{\text{memory}}  \}_{x=0}^{N-1} $ is post-measurement isomorphic to the the set of states $\{P_S^j\}_j$ system alone. That is, for every transformation $\epsilon_S$ that you could apply  to the system before the measurement, there is a corresponding transformation $\epsilon_{SM}$ acting on the $\{P_{SM}^l\}_l$ that is operationally identical. By this we mean that an outside observer would not be able to tell if they are operating with $\epsilon_S$ on a single system before the measurement, or with  $\epsilon_{SM}$ on system and memory after the measurement. In particular, if $\epsilon_S$ is itself another measurement on $S$ within a probabilistic theory, it should yield the same statistics as post-measurement $\epsilon_{SM}$. 
For a quantum example that helps clarify these notions, consider $S$ to be a qubit initially in an arbitrary state $\alpha \ket0_S + \beta \ket1_S$. An agent Alice measures $S$ in the $Z$ basis and stores the outcome in her memory $A$.  While she has a subjective experience of seeing only one possible outcome, an outside observer Bob could model the joint evolution of $S$ and $A$ as
$$ \left( \alpha \ket0_S + \beta \ket1_S\right) \otimes \ket0_A \ \to \   \alpha \ket0_S \ket0_A + \beta \ket1_S \ket1_A. $$
Suppose now that (before Alice's measurement) Bob was interested in performing an $X$ measurement on $S$. This would have been a measurement with projectors $\{ \proj ++_S, \proj --_S\}$, where $ \ket \pm_S = \frac1{\sqrt2} (\ket0_S \pm \ket1_S)$. 
However, he arrived too late: Alice has already performed her $Z$ measurement on $S$. If now Bob simply measured $X$ on $S$ he would obtain uniform statistics, which would be uncorrelated with the initial state of $S$. So what can he do? It may not be very friendly, but he can measure $S$ and Alice's memory $A$ jointly, by projecting onto
\begin{align*}
    \ket+_{SA} &= \frac1{\sqrt2} (\ket0_S \ket0_A + \ket1_S \ket1_A) \\ \ket-_{SA} &= \frac1{\sqrt2} (\ket0_S \ket0_A - \ket1_S \ket1_A),
\end{align*}
which yields the same statistics of Bob's originally planned measurement on $S$, had Alice not measured it first.
This equivalence should also hold in the more general case where the observed system may have been previously correlated with some other reference system: such correlations should be preserved in the measurement process, as modelled from the ``outside'' observer Bob. 

There are many options to formalize this notion that ``every way that an outside observer could have manipulated the system before the measurement, he may now manipulate a subspace of `system and observer's memory,' with the same results.'' 
A possible simplification to restrict our options is to take subsystems and the tensor product structure as primitives of the theory, which is the case for GPTs \cite{Barrett07} but not for general physical theories (like field theories; for a discussion see \cite{Kraemer2018}). In the interest of time, we will for now restrict ourselves to this case, and leave a more general formulation of this condition as future work.   
For simplicity, we also restrict ourselves to information-preserving measurements (excluding for now those where some information may have leaked to an  environment external to Alice's memory), which are sufficient to derive the contradiction.

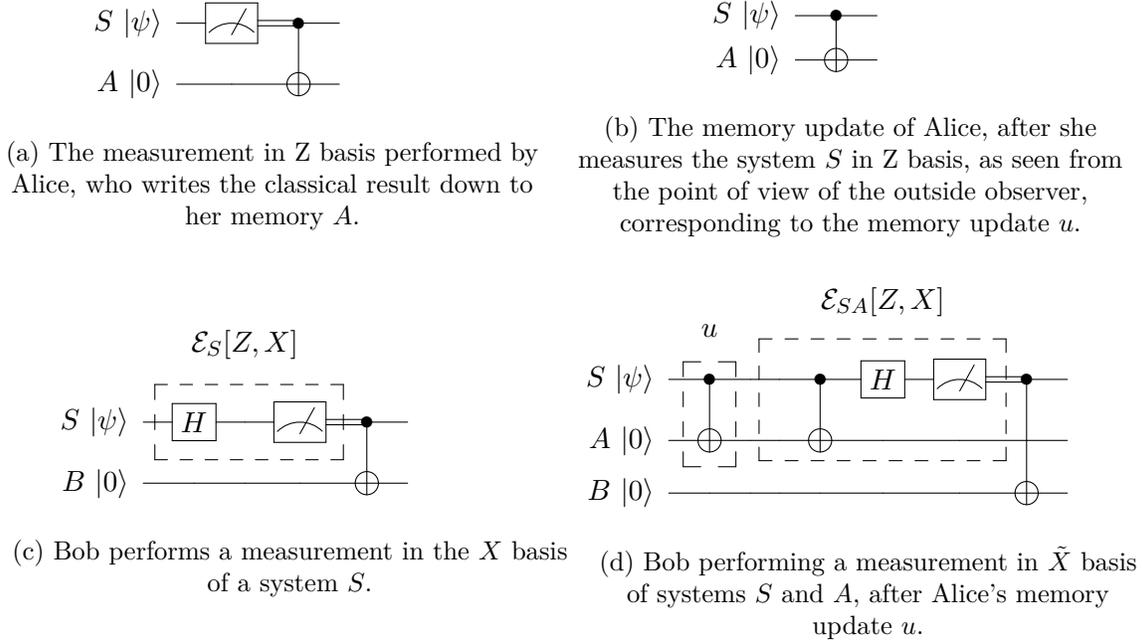
\begin{figure}[t!]
\centering
    \begin{subfigure}{0.45\textwidth}
        \begin{align*}
        \Qcircuit @C=1em @R=1em {
        \lstick{S \ \ket{\psi}} & \meter & \control{1} \cw & \qw & \\
        \lstick{A \ \ket{0}} & \qw & \targ \qwx[-1] & \qw \\
        }
        \end{align*}
        \caption{The measurement in Z basis performed by Alice, who writes the classical result down to her memory $A$.}
        \label{fig:view_of_alice}
    \end{subfigure}
    \
    \begin{subfigure}{0.45\textwidth}
        \begin{align*}
        \Qcircuit @C=1em @R=1em {
        \lstick{S \ \ket{\psi}} & \ctrl{1} & \qw & \\
        \lstick{A \ \ket{0}} & \targ & \qw & \\
        }
        \end{align*}
        \caption{The memory update of Alice, after she measures the system $S$ in Z basis, as seen from the point of view of the outside observer, corresponding to the memory update $u$.}
        \label{fig:view_on_alice}
    \end{subfigure}
    \newline
    \begin{subfigure}{0.45\textwidth}
        \begin{align*}
        \Qcircuit @C=1em @R=1em {
        && \mbox{$\mathcal{E}_{S}[Z,X]$} &&&\\
        &&&&&& \\
        \lstick{S \ \ket{\psi}}  & \gate{H} & \qw & \meter & \control{1} \cw & \qw \\
        \lstick{B \ \ket{0}} & \qw & \qw& \qw & \targ \qwx[-1] & \qw \gategroup{3}{2}{3}{4}{1.2em}{--}
        }
        \end{align*}
        \caption{Bob performs a measurement in the  $X$ basis of a system $S$.}
        \label{fig:operation_s}
    \end{subfigure}
    \
    \begin{subfigure}{0.45\textwidth}
        \begin{align*}
        \Qcircuit @C=1em @R=1em {
        &&&&&  \mbox{$\mathcal{E}_{SA}[Z, X]$} &&& \\
        & \mbox{$u$} & & & & & & & \\
        \lstick{S \ \ket{\psi}} & \ctrl{1} & \qw & \qw & \ctrl{1} & \gate{H} & \meter & \control{1} \cw & \qw \\
        \lstick{A \ \ket{0}} & \targ & \qw & \qw & \targ & \qw & \qw & \qw & \qw \\
        \lstick{B \ \ket{0}} & \qw & \qw & \qw & \qw & \qw & \qw & \targ \qwx[-2] & \qw \gategroup{3}{4}{4}{7}{1.4em}{--} \gategroup{3}{2}{4}{2}{1em}{--}
        }
        \end{align*}
        \caption{Bob performing a measurement in $\tilde X$ basis of systems $S$ and $A$, after Alice's memory update $u$.}
        \label{fig:operation_sa}
    \end{subfigure}
\caption{{ \bf  The measurement and memory update in quantum theory from different perspectives.} 
From Alice's point of view, the measurement of the system $S$ either in Z basis yields a classical result, which she records to her memory $A$, performing a classical CNOT (Figure \ref{fig:view_of_alice}). From an outside observer, Bob's perspective, as he is not aware of Alice's measurement result, the CNOT is a quantum entangling operation, which corresponds to the memory update $u$ (Figure \ref{fig:view_on_alice}). If he had access to the system $S$ prior to the measurement by $A$, and wanted to measure it in X basis ($\{\ket{+}_S,\ket{-}_S\}$), he would have to perform an operation $\mathcal{E}_S [Z,X]$ (and then copy the classical result into his memory $B$) (Figure \ref{fig:operation_s}). If the system $S$ was initially in a state $\ket{\psi} = \ket{+}_S$, then a proposition which would correspond to this operation is $\phi[\mathcal{E}_S [Z,X](\ket{\psi}_S)] = ``s=+"$. However, if the measurement in Z is already performed by $A$ and the result is written to her memory, the whole process described by Bob as a memory update $u$, and in order to comply his initial wish to measure $S$ only, he can perform an operation $\mathcal{E}_{SA}[Z,X]$ on $S$ and $A$ together instead, which is a measurement in $\{\ket{+}_{SA},\ket{-}_{SA}\}$ basis (Figure \ref{fig:operation_sa}). A proposition which this operation yields is $\phi[\mathcal{E}_{SM_i}\circ u(\ket{\phi}_S)] = ``sa=+"$ (as $\ket{\xi}_{SA} = \ket{+}_{SA}$), which naturally follows from $``s=+"$, given the structure of the memory update $u$.}
\label{fig:circuits_views}
\end{figure}

\begin{definition}[Information-preserving memory update]
\label{def:memory_update}
Let $\{P_S^j\}_j$ be a set of states of a system $S$ that is being studied by an agent $A_i$ with a memory $M_i$, and $\{P_{SM_i}^l\}_l$ be a set of states of the joint system $SM_i$, which consists of the systems $S$ and $M_i$. Then a map $u: \{P_S^j\}_j \rightarrow \{P_{SM_i}^l\}_l$ is called an \emph{information-preserving memory update} if for all operations $\mathcal{E}_S: \{P_S^j\}_j \rightarrow \{P_S^j\}_j$ on the system $S$, there exists an operation $\mathcal{E}_{SM_i}: \{P_{SM_i}^l\}_l \rightarrow \{P_{SM_i}^l\}_l$ such that:
\begin{gather*}
\forall P_S^i, \forall A_j\in\mathcal{A} \quad K_j \phi[\mathcal{E}_S(P_S^i)] \Rightarrow K_j \phi[\mathcal{E}_{SM_i}\circ u(P_S^i)].
\end{gather*}
\end{definition}

See Figure \ref{fig:circuits_views} for an example. 
In general, the memory update map $u$ need not be reversible; for example, in box world it is an irreversible transformation, as we will see later. 

Note that the characterization of measurements introduced in this section is rather minimal. In physical theories like classical and quantum mechanics, measurements have other natural properties that we do not require here. Two striking examples are ``after her measurement, Alice's memory becomes correlated with the system measured in such a way that, for any subsequent operation that Bob could perform on the system, there is an equivalent operation he may perform on her memory'' and ``the correlations are such that there exists a joint operation on the system and Alice's memory that would allow Bob to conclude which measurement Alice performed.'' While these properties hold in the familiar classical and quantum worlds, we do not know of other physical theories where measurements can satisfy them, and they require Bob to be able to act independently on the system and on Alice's memory, which may not always be possible. For example, we will see that in box world, these two subsystems become \emph{superglued} after Alice's measurement, and that Bob only has access to them as a whole and not as individual components.  As such, we will not require these properties out of measurements, for now. We revisit this discussion in Section~\ref{sec:conclusions}.

\section{Box world: states and memories}
\label{sec:boxworld}
 Generalised probabilistic theories \cite{Hardy01, Barrett07} (GPTs) provide an an operational framework for describing probabilistic theories, including classical and quantum theories where the physical systems are taken as black boxes, characterized only by their input and output behaviour. The \emph{state} of a system is represented by a probability vector ${\bf P}$ that encodes the probabilities of possible outcomes given all the possible choices of measurement. This is a single-shot characterization of a system: the post-measurement state can be represented by a new probability vector, and the update rules depend on the specific theory. 
 
 In this paper, we employ the framework for information processing in GPTs presented by Barrett in \cite{Barrett07}, and use we the term ``box world'' to denote the set of theories that Barrett originally calls \emph{Generalised No-Signalling Theories}. We will derive the paradox in box world, which is a particular instance of a GPT. However, the general assumptions proposed in Section~\ref{sec:conditions} can also be applied to more general GPTs that do not obey the standard no signalling principle \cite{Grunhaus1996, Horodecki16} or that which obey different physical principles. We present here the minimal formalism needed to follow the argument; see  Appendix~\ref{appendix:boxes} for more details. 

\subsection{States and operations (review)}
\label{sssec: statesop}

\paragraph{Individual states.} The so-called generalised bit or \emph{gbit} is a system completely characterized by two binary measurements which can be performed on it~\cite{Barrett07}. Such sets of measurements that completely characterise the state of a system are known as \emph{fiducial measurements}. The state of a gbit is thus fully specified by the vector
\begin{equation}
\label{eq: gbit}
    \vec{P}_{gbit}=\left(
\begin{array}{c}
 P(a=0|X=0)\\
 P(a=1|X=0)\\
\hline
 P(a=0|X=1)\\
 P(a=1|X=1)\\
\end{array}
\right),
\end{equation}
where $X=0$ and $X=1$ represent the two choices of measurements and $a \in \{0,1\}$ are the possible outcomes (Figure~\ref{fig:gbit}).
Analogously, a classical bit is a system characterized by a single binary fiducial measurement, 
\begin{equation}
\label{eq: bit}
    \vec{P}_{bit}=\left(
\begin{array}{c}
 P(a=0|X=0)\\
 P(a=1|X=0)
\end{array}
\right),
\end{equation}

and, in quantum theory, a qubit is characterized by three fiducial measurements (corresponding, for example, to three directions $X$, $Y$ and $Z$ in the Bloch sphere), 
\begin{equation}
\label{eq: qubit}
    \vec{P}_{qubit}=\left(
\begin{array}{c}
 P(a=0|X=0)\\
 P(a=1|X=0)\\
\hline
 P(a=0|X=1)\\
 P(a=1|X=1)\\
 \hline
 P(a=0|X=2)\\
 P(a=1|X=2)\\
\end{array}
\right).
\end{equation}

\begin{figure}[t]
    \centering
    \begin{subfigure}{0.35\textwidth}
    \vspace{1cm}
    \includegraphics[width=\textwidth]{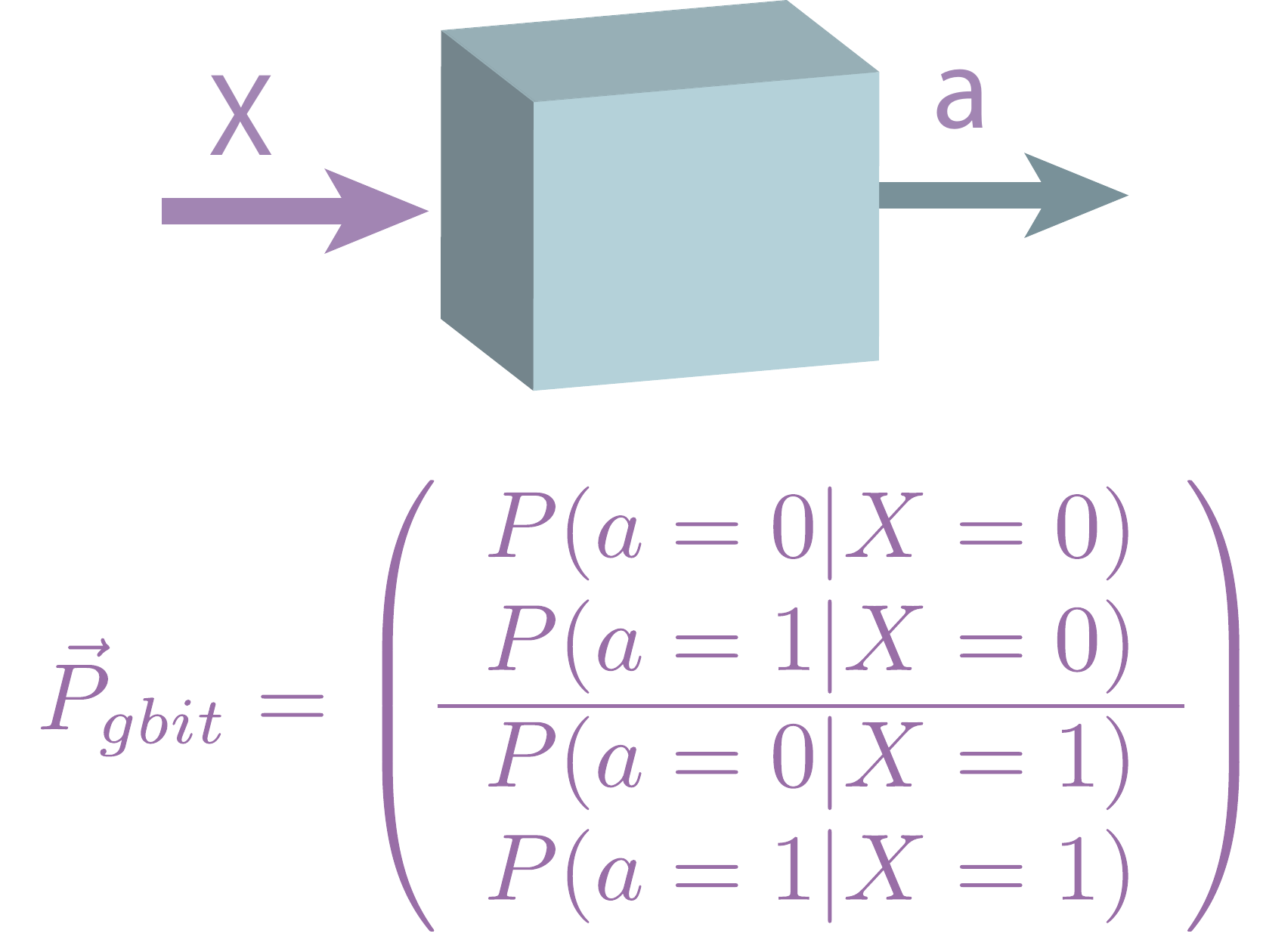}
    \caption{\textbf{G-bit.} A gbit is a function with binary input and output, characterized by the  probability vector $\vec P_{gbit}$, also called the state vector. }
    \label{fig:gbit}
    \end{subfigure}
    \quad 
\begin{subfigure}{0.6\textwidth}    \centering
    \includegraphics[width=\textwidth]{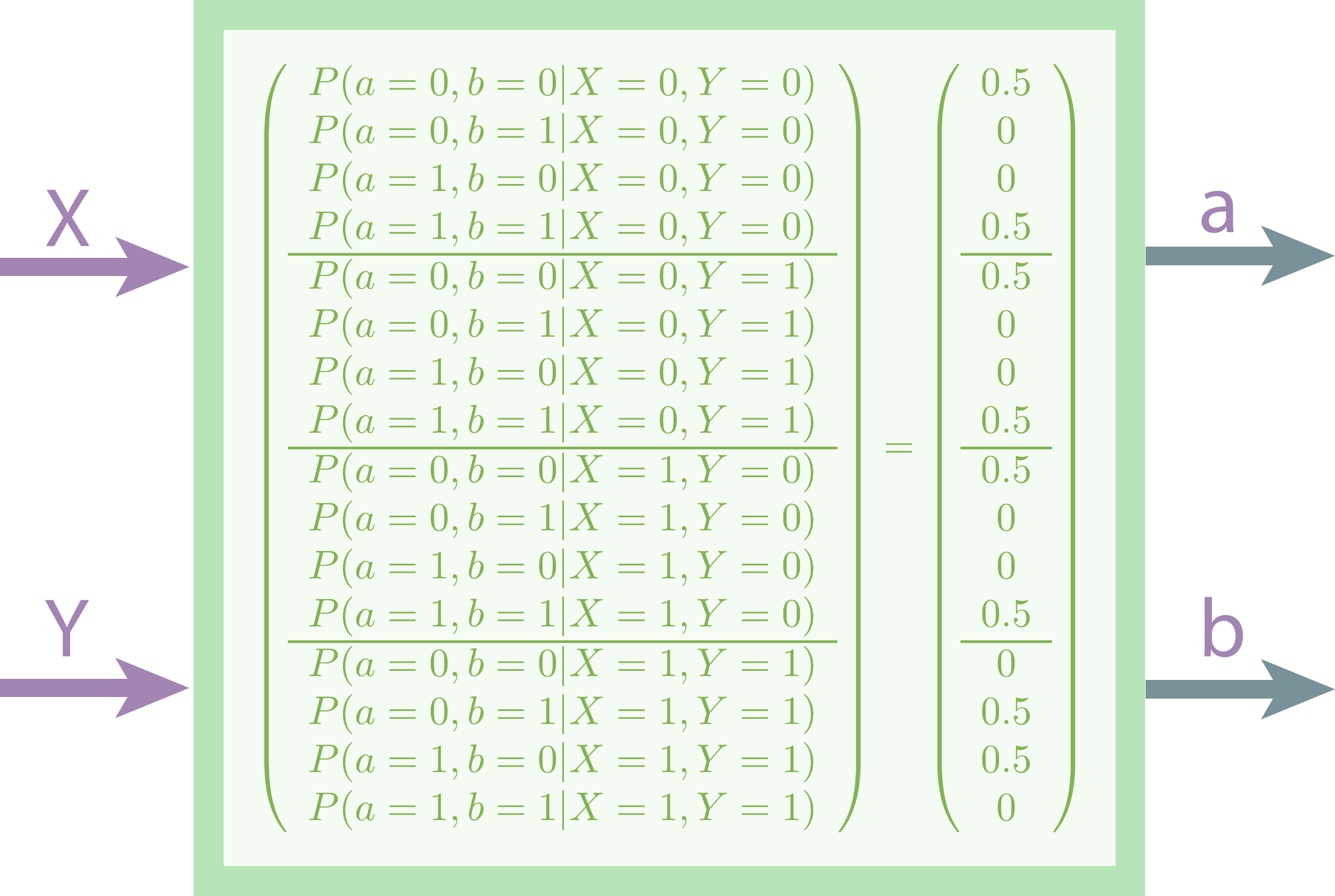}
    \caption{\textbf{PR box.} The PR box has two binary inputs $X, Y$ and two binary outputs $a,b$, satisfying $XY = a\oplus b$, and otherwise uniformly random (state vector on the right). Usually it is applied in the context of two space-like separated agents, each providing one of the inputs and obtaining the respective output. The box is non-signalling, and maximally violates the CHSH inequality~\cite{Popescu1994}.}
    \label{fig:PRbox}
\end{subfigure}

\caption{{\bf Boxes in Generalized Probabilistic Theories.} The modular objects of GPTs are input/output functions depicted as boxes and characterized by probability vectors. Each function (or box) can be evaluated once, and it may or not correspond to a physical system being probed; even if it is, nothing is assumed about the post-evalutation state of the system (unlike quantum theory, which specifies the post-measurement state of a system given its initial state and the measurement device). }
\label{fig:GTP_states}
\end{figure}

For normalized states, we have $|\vec{P}|=\sum_i P(a=i|X=j)=1,   \forall\, j$.
The set of possible states of a gbit is convex, with extremes
\begin{align}
\label{eq: gbitpure}
   \vec{P}_{00}=\left(
\begin{array}{c}
1\\
0\\
\hline
1\\
0\\
\end{array}
\right), \qquad
 \vec{P}_{01}=\left(
\begin{array}{c}
1\\
0\\
\hline
0\\
1\\
\end{array}
\right), 
\quad
   \vec{P}_{10}=\left(
\begin{array}{c}
0\\
1\\
\hline
1\\
0\\
\end{array}
\right),
\qquad
 \vec{P}_{11}=\left(
\begin{array}{c}
0\\
1\\
\hline
0\\
1\\
\end{array}
\right).
\end{align}
These correspond to pure states. In the qubit case, the extremes correspond to all the points on the surface of the Bloch sphere, for example
\begin{equation}
    \vec{P}_{\ket+}=\left(
\begin{array}{c}
 1\\
 0\\
\hline
 \nicefrac12 \\
\nicefrac12 \\
 \hline
 \nicefrac12 \\
 \nicefrac12 \\
\end{array}
\right), 
\qquad
    \vec{P}_{\ket-}=\left(
\begin{array}{c}
 0\\
 1\\
\hline
 \nicefrac12 \\
\nicefrac12 \\
 \hline
 \nicefrac12 \\
 \nicefrac12 \\
\end{array}
\right), 
\qquad
    \vec{P}_{\ket0}=\left(
\begin{array}{c}
 \nicefrac12 \\
 \nicefrac12 \\
\hline
 \nicefrac12 \\
 \nicefrac12 \\
 \hline
 1\\
 0\\
\end{array}
\right)\qquad
    \vec{P}_{\ket1}=\left(
\begin{array}{c}
 \nicefrac12 \\
 \nicefrac12 \\
\hline
 \nicefrac12 \\
 \nicefrac12 \\
 \hline
 0\\
 1\\
\end{array}
\right).
\end{equation}
Note that in box world, pure gbits are deterministic for both alternative measurements, whereas in quantum theory at most one fiducial measurement can be deterministic for each pure qubit, as reflected by uncertainty relations. 
We denote the set of allowed states of a system $A$ by $\mathcal{S}^A$.

\paragraph{Composite states.} The state of a bipartite system $AB$, denoted by $\vec{P}^{AB} \in \mathcal{S}^{AB}$ can be written in the form $\vec{P}^{AB}=\sum_i r_i\ \vec{P}^A_i\otimes \vec{P}^B_i$ where $r_i$ are real coefficients\footnote{Note that it is not necessary that the coefficients $r_i$ be positive and sum to one. If this is the case, then the composite state would be separable and hence local, otherwise, the state is entangled \cite{Barrett07}.} and $\vec{P}^A_i\in \mathcal{S}^A$, $\vec{P}^B_i\in \mathcal{S}^B$ can be taken to be pure and normalised states of the individual systems $A$ and $B$~\cite{Barrett07}. Thus, a general 2-gbit state $\vec{P}^{AB}_2$ can be written as in Figure~\ref{fig:PRbox} (left), where $X, Y\in \{0,1\}$ are the two fiducial measurements on the first and second gbit and $a,b \in \{0,1\}$ are the corresponding measurement outcomes. The PR box $\vec{P}_{PR}$, on the right, is an example of such a 2 gbit state that is valid in box world, which satisfies the condition $a\oplus b = x y$~\cite{Popescu1994}.

\paragraph{State transformations.} Valid operations are represented as matrices that transform valid state vectors to valid state vectors (Appendix~\ref{appendix:boxes}). In addition, we only have access to the (single-shot) input/output behaviour of systems, so in practice all valid operations in box world take the form of classical wirings between boxes, which correspond to pre- and post-processing of input and output values, and convex combinations thereof~\cite{Barrett07}.  
For example, 
 bipartite joint measurements on a 2-gbit system can be decomposed into convex combinations of classical ``wirings'', as shown in  Figure~\ref{fig: bipartitemeas}.
In contrast, quantum theory allows for a richer structure of bipartite measurements by allowing for entangling measurements (e.g.\ in the Bell basis), which cannot be decomposed into classical wirings. 
  Bipartite transformations on multi-gbit systems turn out to be  classical wirings as well \cite{Barrett07}. Reversible operations in particular  consist only of trivial wirings: local operations and permutations of systems \cite{Gross2010}. 
 One cannot perform entangling operations such as a coherent copy (the quantum CNOT gate) \cite{Barrett07, Short2006}, which is required in the original version of the Frauchiger-Renner experiment.
\begin{figure}[t]
    \centering
\begin{tikzpicture}

\draw[arrows={-stealth}] (2.2,0)--(2.2,-1); \draw[arrows={-stealth}] (6.2,0)--(6.2,-1);
\node[align=center] at (2.5,-0.5) {$X$}; \node[align=center] at (6.5,-0.5) {$Y$};
\draw (1.6,-2) rectangle node[align=center]{$\mathbf{A}$} (2.8,-1);
\draw (5.6,-2) rectangle node[align=center]{$\mathbf{B}$} (6.8,-1); 
\draw (5.3,0) rectangle node[align=center]{$Y=f_1(a)$} (7.1,1);\draw[dashed] (2.8,-1.5)--(5.6,-1.5); \draw[arrows={-stealth}] (6.2,-2)--(6.2,-3); \node[align=center] at (6.5,-2.5) {$b$}; \draw (5.6,-4.2) rectangle node{$o=f_2(a,b)$} (8.6,-3);
 \draw[arrows={-stealth}] (7.7,-2)--(7.7,-3);   \node[align=center] at (8,-2.5) {$a$};  
\draw[arrows={-stealth}] (7.1,-4.2)--(7.1,-5.2);  \node[align=center] at (7.1,-5.5) {$o$}; \draw (2.2,-2)--(2.6,-2.5); \draw (2.6,-2.5) to [out=315,in=135] (5.8,1.5); \draw[arrows={-stealth}] (5.8,1.5)--(6.2,1);
\node[align=center] at (2.2,-2.4) {$a$};
\end{tikzpicture}
    \caption{{\bf Bipartitite measurements in boxworld.} Any bipartite measurement on a 2-gbit box world system can be decomposed into a procedure (or convex combinations thereof) of the following form. Alice first performs a measurement $X$ on one of the gbits (labelled $A$), and forwards the outcome $a$ to Bob. Bob then performs a measurement $Y=f_1(a)$, which may depend on $a$, on the other gbit (labelled $B$), obtaining the outcome $b$. The final measurement outcome $o$ of the joint measurement can be computed by Bob as a function $f_2(a,b)$. All allowed bipartite measurements are convex combinations of this type of classical wirings~\cite{Barrett07}.}
    \label{fig: bipartitemeas}
\end{figure}
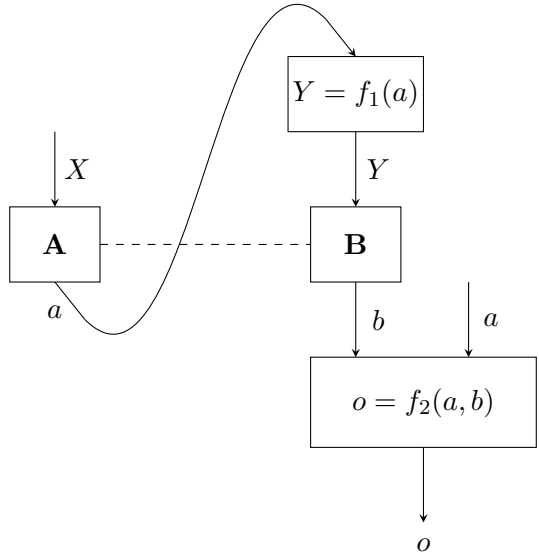

\subsection{Agents, memory and measurement in box world}

We will now instantiate our general conditions for agents, memories and measurements (definitions \cref{def:measurements,def:agents,def:memory_update}) in box world. As there is no physical theory for the dynamics behind box world, there is plenty of freedom in the choice of implementation. In principle each such choice could represent a different physical theory leading to the same black-box behaviour in the limit of a single agent with an implicit memory. This is analogous to the way in which different versions of quantum theory (Bohmian mechanics, collapse theories, unitary quantum mechanics with von Neumann measurements) result in the same  effective theory in that limit.

\begin{definition}[Agents in box world]
\label{def:agentsbox}
Let $\mathbb T$ be the theory that describes box world, according to~\cite{Barrett07}. 
As per definition~\ref{def:agents}, an agent $A_i \in\mathcal{A}$ is described by a knowledge operator $K_i\in\mathcal{K_\mathcal{A}}$ and a physical memory $M_i\in\mathcal{M_\mathcal{A}}$. 

We will focus on the case where the memory consists of bit or gbits. Each agent may study other systems according to the theory $\mathbb T$. An agent's memory records the results and the consequences of the studies conducted by them, and may be an object of a study by other agents.
\end{definition}

It is worth mentioning that boxes do not correspond to physical systems, but to input/output functions that can only be evaluated once. As such, the post-measurement state of a physical system is described by a whole new box. The notion of an individual system itself, as we will see, may be unstable under measurements --- some measurements \emph{glue} the system to the observer's memory, in a way that makes individual access to the original system impossible.

\begin{figure}
    \centering
    \includegraphics[width=0.5\textwidth]{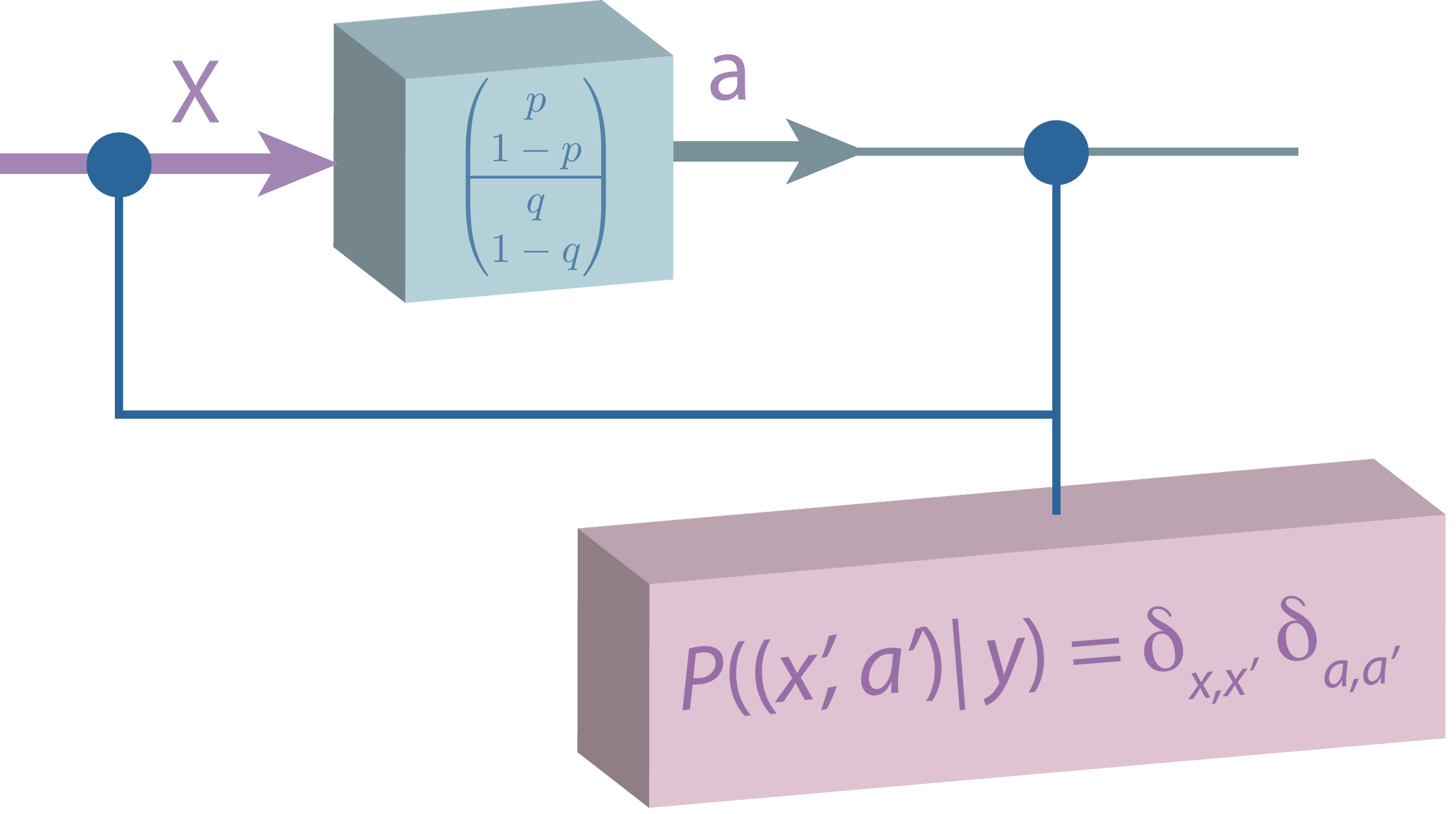}
    \caption{\textbf{Measurement: observer's perspective.} An agent Alice measures a system with measurement setting $X$, and obtains outcome $a$ with a given probability. In the language of GPTs this corresponds to running the box that encodes the measurement statistics. Alice may then the measurement data (input and output) to memory. If this is a classical memory, like a notebook, the procedure corresponds to preparing a new box (to be run later by herself), which outputs the pair $(X,a)$ deterministically.}
    \label{fig:measurement_Alice}
\end{figure}

\paragraph{Measurement: observer's perspective.} From the point of view of the observer who is measuring (say Alice), making a measurement on a system corresponds simply to running the box whose state vector encodes the measurement statistics. 
Alice may then commit the result of her measurement to a physical memory, like a notebook where she writes `I measured observable $X$ and obtained outcome $a$.' To be useful, this should be a memory that may be consulted later, i.e.\ it could receive an input $Y=$`open and read the memory', and output the pair $(X,a)$. In the language of GPTs, this means that Alice, from her own perspective, prepares a new box with one input $Y$ and two outputs $(X',a')$, with the behaviour $\vec P((X',a')|Y) = \delta_{X, X'} \delta_{a,a'}$, which depends on her observations (Figure~\ref{fig:measurement_Alice}). She may later run this box (look at her notebook) and recover the measurement data. The exact dimension of the box will depend on how Alice perceives and models her own memory; for example it could consist of two bits, or two gbits, or, if we think that before the measurement she stored the information about the choice of observable elsewhere, it could be a single bit or gbit encoding only the outcome. We leave this open for now, as we do not want to constrain the theory too much at this stage. 

\begin{figure}[t!]
    \centering
   \begin{subfigure}{0.45\textwidth}
    \includegraphics[width=\textwidth]{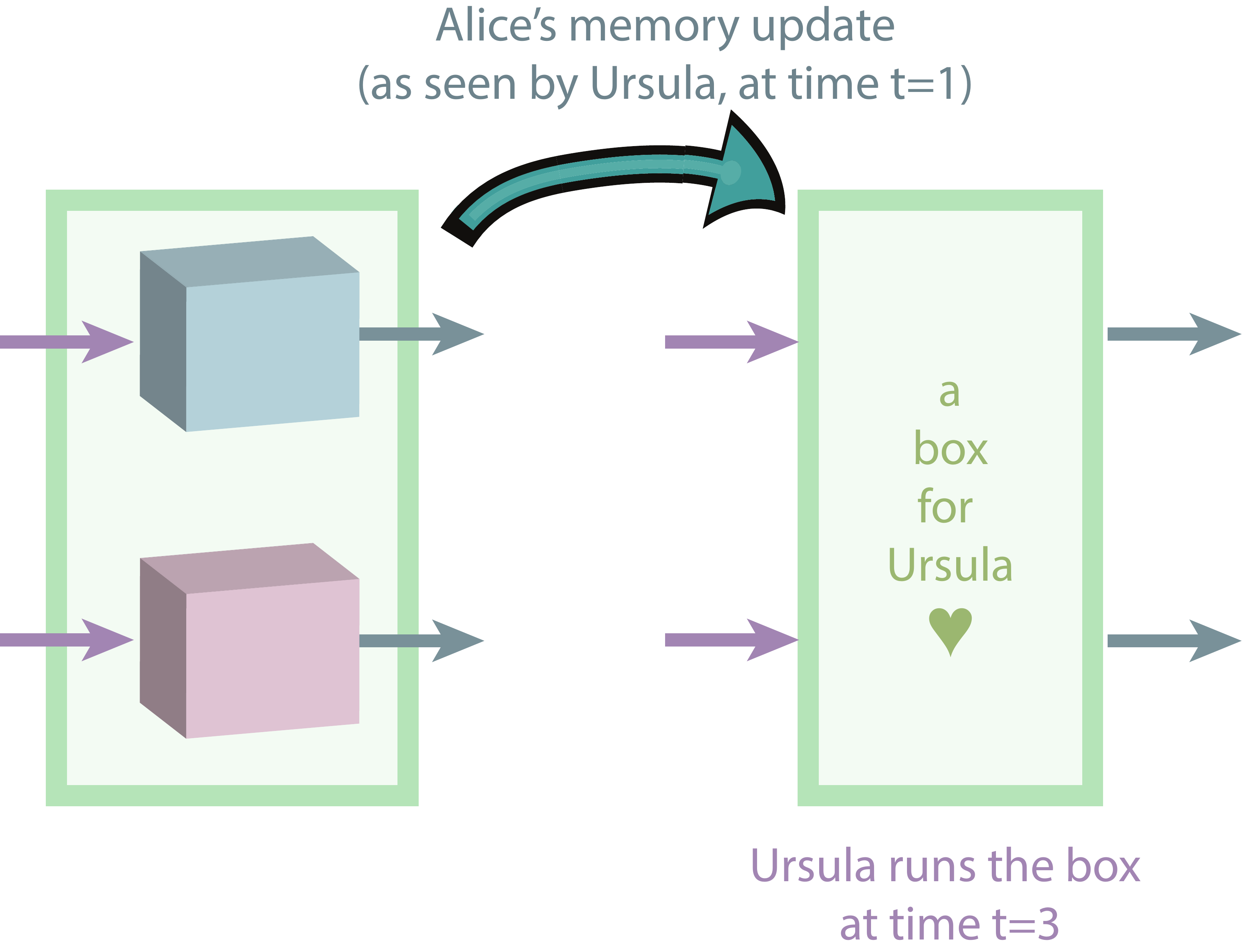}
    \caption{Generally, in GPTs with some notion of subsystems, Ursula can think of the physical system measured by Alice, and Alice's memory pre-measurement as two boxes, which Ursula could in principle run if Alice chose not to measure (left).  From Ursula's perspective, Alice's measurement corresponds to some transformation that results on a final state on which Ursula can later act. This final state can be represented by a new box available to Ursula, which will have in principle a different behaviour, depending on the concrete physical theory.}
    \label{fig:update_general}
\end{subfigure}
\qquad
\begin{subfigure}{0.45\textwidth}
    \includegraphics[width=\textwidth]{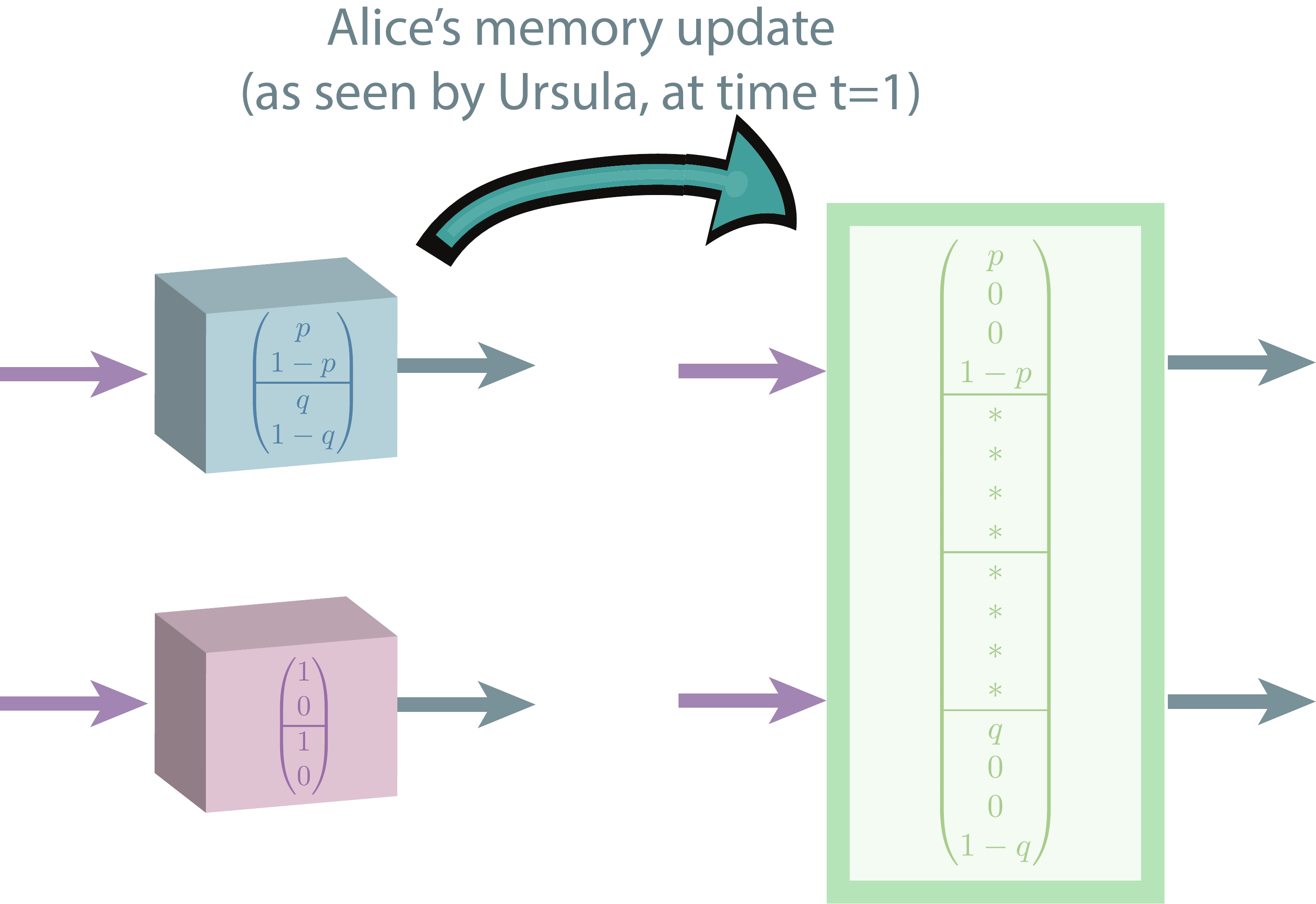}
    \caption{In  box world, if the two initial systems correspond to small gbit boxes, and Alice's memory is initialized as shown, and if we want to preserve the global system dimensions, then the rules for allowed transformations limit the statistics of Ursula's final box to be of the form shown in the right (Appendix~\ref{appendix:proofs}). The asterisks represent arbitrary values, which will depend on the choice of implementation of Alice's measurement. This transformation is in principle non-reversible: note that in the final box, Ursula cannot address system and memory independently, but only the global, \emph{superglued} box.}
    \label{fig:update_conditions}

\end{subfigure}
\caption{{\bf Memory update after a measurement: an outsider's perspective.} Here Alice makes a measurement of a system  (blue, top) at time $t=1$ and stores her outcome in her memory (pink, bottom). The question is how an outsider, Ursula, models Alice's measurement. In particular, what can Ursula do with the post-measurement state?}
    \label{fig:update_basics}
\end{figure}

\paragraph{Measurements: inferences.}
To see the kind of inferences and conclusions that an agent can take from a measurement in box world, it's convenient to look at the example where Alice and Bob share a PR box. Suppose that Alice measured her half of the box with input $X=1$ and obtained outcome $a=0$. From the PR correlations, $XY=a\oplus b$, she can conclude that if Bob measures $Y=0$, he must obtain $b=0$, and if he measures $Y=1$, he must obtain $b=1$. This is independent  \\
\begin{figure}[H]
    \centering
   \begin{subfigure}{0.4\textwidth}
    \includegraphics[width=\textwidth]{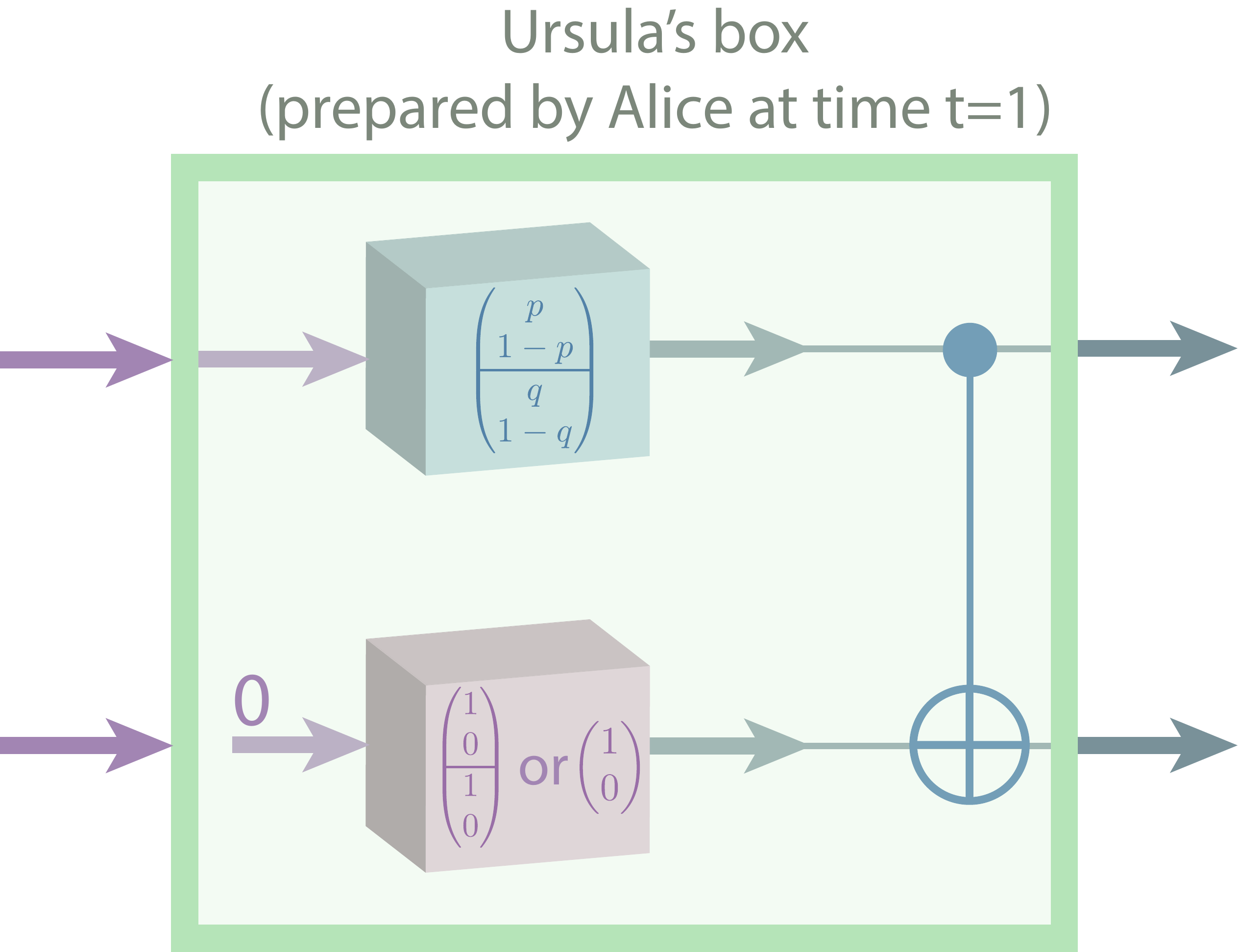}
    \caption{From Ursula's perspective, Alice has not yet run the boxes corresponding to the system measured and her memory; she simply wired the outputs of the two boxes with a controlled-{\sc{not}} gate, so that  the measurement output is copied to the output of the memory. This is analogous to the quantum case, where from Ursula's view Alice has not performed a projective measurement, but simply entangled system and memory.
   When Ursula later runs the outer green box, she provides  two inputs, which go through the circuit shown, resulting in two identical outputs.}
    \label{fig:update_copy}
\end{subfigure}
\qquad
\begin{subfigure}{0.5\textwidth}
    \includegraphics[width=\textwidth]{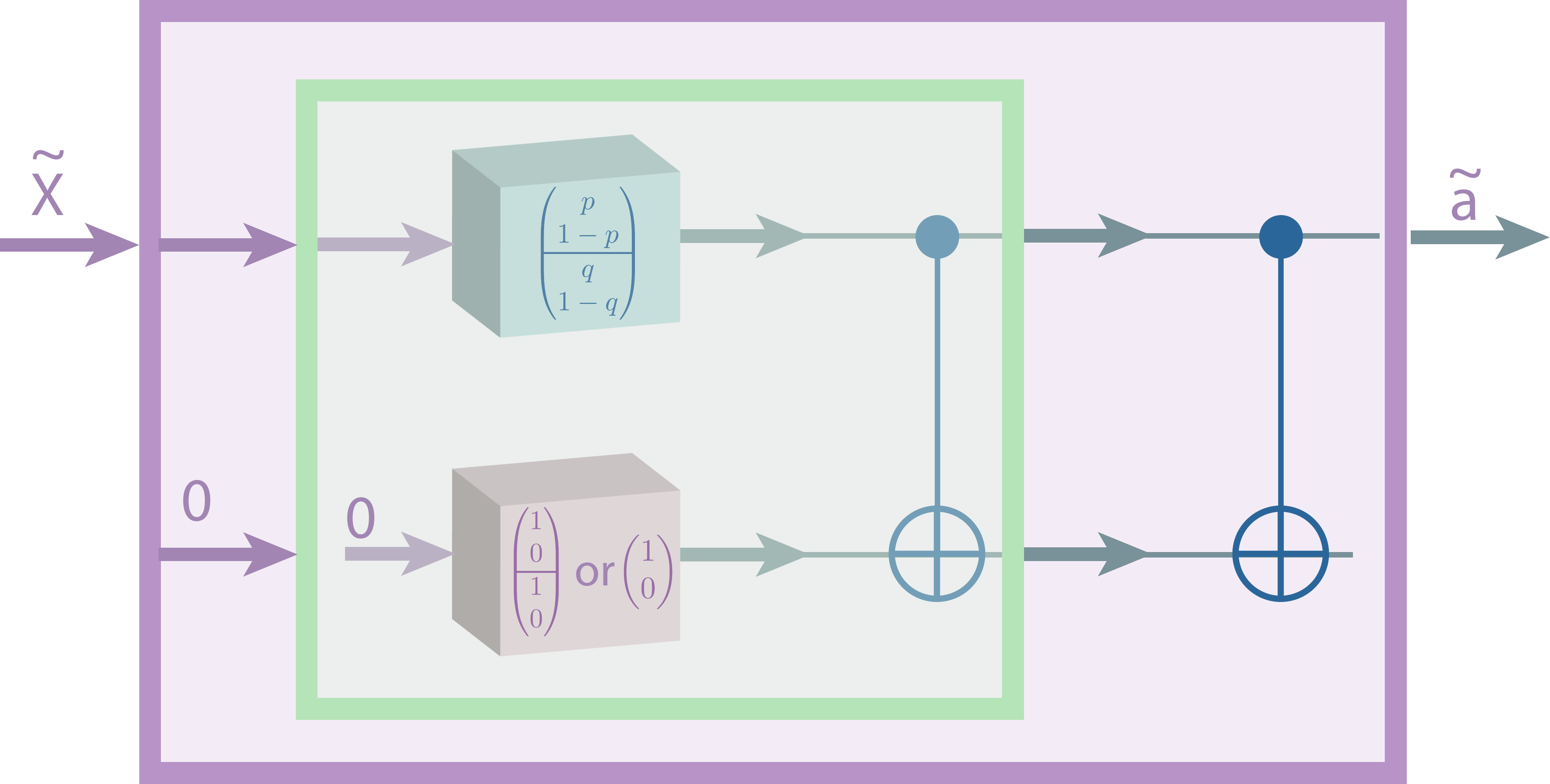}
    \caption{In order for the model to be information-preserving, we need Ursula to be able to do some pre- and post-processing (outer pink box), such that the final box has the same behaviour as the initial state of the system measured by Alice (inner blue box on top). This is achieved, for example, by Ursula fixing her second input to 0, and undoing the controlled-{\sc{not}} gate at the end, discarding the second (trivial) output. The result is a box with binary input $\tilde X$ and binary output $\tilde a$, which has the desired behaviour. This property carries on to bipartite scenarios where Alice measures half of a joint state.}
    \label{fig:update_ursula}
\end{subfigure}
\caption{\textbf{Information-preserving memory update.} This (trivial) physical implementation of Alice's measurement in box world satisfies the conditions of Figure~\ref{fig:update_conditions} and is information-preserving, in the sense that an external agent, Ursula, can run the final box as if it were the original, pre-measurement state of the system that Alice measured, in analogy to the quantum case (Figure~\ref{fig:circuits_views}). The crucial detail is that Ursula is not allowed to open her box (in green) and access the circuitry inside.  Note that there are other possibilities for modelling measurements --- this is the simplest one that still allows us to derive the paradox. For example, the choice of keeping two binary inputs in Ursula's box and discarding the second one (replacing it with 0) is an arbitrary one, picked for simplicity. Details and proofs in Appendix~\ref{appendix:proofs}.}
    \label{fig:update_implementation}
\end{figure}

\noindent of whether Bob's measurement happens before or after Alice (or even space-like separated).
She could reach similar deterministic conclusions for her other choice of measurement and possible outcomes. 
In the language of Definition~\ref{def:measurements}, we have 
\begin{align*}
    \phi_{X=0, a=0} &= ``[Y=0 \implies b=0] \wedge [Y=1 \implies b=0]", \\
    \phi_{X=0, a=1} &= ``[Y=0 \implies b=1] \wedge [Y=1 \implies b=1]", \\
    \phi_{X=1, a=0} &= ``[Y=0 \implies b=0] \wedge [Y=1 \implies b=1]", \\
    \phi_{X=1, a=1} &= ``[Y=0 \implies b=1] \wedge [Y=1 \implies b=0]".
\end{align*}

\paragraph{Measurement: memory update from an outsider's perspective.}
Next we need to model how an outsider agent, Ursula, models Alice's measurement, in the case where Alice does not communicate her outcome to Ursula.\footnote{Naming convention: as we will see in Section~\ref{sec:paradox}, in the proposed experiment we have two ``internal'' agents, Alice and Bob, who will in turn be measured by two ``external'' agents, Ursula and Wigner respectively. Ursula is named after Le Guin. In the example of Section~\ref{sec:conditions} the internal agent was Alice and the external Bob, so that their different pronouns could help keep track of whose memory we were referring to, but we trust that the reader has got a handle on it by now.} 
Suppose that all agents share a time reference frame, and Alice makes her measurement at time $t=1$. From Ursula's perspective, in the most general case, this will correspond to Alice preparing a new box, with some number of inputs and outputs, which Ursula can later run (Figure~\ref{fig:update_general}). The exact form of this box will depend on the underlying physical theory for measurements: in the quantum case it corresponds to a box with the measurement statistics of a state that's entangled between the system measured and Alice's memory, as we saw. In classical mechanics, it will correspond to perfect classical correlations between those two subsystems. In a theory of very destructive measurements, it could be that Alice's post-measurement state is trivial from the point of view of Ursula and the resulting box is void. Now suppose that we would like to have a physical theory where the dimension of systems is preserved by measurements: for example, if the system that Alice measures is instantiated by a box with binary input and output (e.g.~a gbit, or half of a PR-box), and Alice's memory, where she stores the outcome of the measurement (as in Figure~\ref{fig:measurement_Alice}) is also represented as a gbit, then we would want the post-measurement box accessible to Ursula to have in total two binary inputs and two binary outputs (or more generally, four possible inputs and four possible outputs). Note that this is not a required condition for a theory to be \emph{physical} per se --- it is just a familiar rule of thumb that gives some persistent meaning to the notion of subsystems and dimensions. In such a theory that supports box world correlations, we find that the allowed statistics of Ursula's box must satisfy the conditions of Figure~\ref{fig:update_conditions} (proof in Appendix~\ref{appendix:proofs}). These conditions still leave us some wiggle room for possible different implementations. 

\paragraph{Measurements: information-preserving memory update.}
In order to find a multi-agent paradox, we will need a model of memory update that is information-preserving, in the sense of Definition~\ref{def:memory_update}.  This does not imply that Alice's transformation (as seen by Ursula) be reversible: in fact, it will \emph{glue} two subsystems such that Ursula will only be able to address them as a whole, but the relevant fact is that Ursula can apply some post-processing in order to obtain a new box with the same behaviour as the pre-measurement system that Alice observed. In Figure~\ref{fig:update_implementation} we give an example of a model that satisfies these conditions, in addition to the conditions of Figure~\ref{fig:update_conditions}. As we foreshadowed, this model is not completely satisfying from a physical standpoint: it looks rather trivial  (a post-processing of classical outputs); the super-gluing is postulated rather than naturally emergent; and, unlike quantum von-Neumann measurements, it does not give Ursula information about the nature of Alice's measurement. It helps to think of it as one minimal implementation among many possible, which already allows us to derive a paradox. We discuss these limitations and alternatives in Section~\ref{sec:discussion}. 
What is important here (and proven in Appendix~\ref{appendix:proofs}) is that this model generalizes to the case where Alice measures half of a bipartite state, like a PR box. That is, suppose that Alice and Bob share a PR box. Imagine that at time $t=1$ Alice makes her measurement $X$, obtaining (from her perspective) an outcome $a$, and that Bob makes his measurement $Y$ at time $t=2$, obtaining outcome $b$. As usual, if Alice and Bob were to communicate at this point, they would find that $XY=a\oplus b$, and indeed the propositions $\phi_{X,a}$ and $\phi_{Y,b}$ that represent their subjective measurement experience would hold. But now suppose that Alice and Bob do not get the chance to communicate and compare their input and outputs; instead, at time $t=3$, an observer Ursula, who models Alice's measurement as in Figure~\ref{fig:update_copy}, runs the box corresponding to Alice's half of the PR box and Alice's memory, and applies the post-processing of Figure~\ref{fig:update_ursula}. Ursula's input is $\tilde X$ and her output is $\tilde a$. Then the claim is that $\tilde XY=\tilde a\oplus b$: that is, Ursula and Bob effectively share a PR box. This is proven in Appendix~\ref{appendix:proofs}. 
We now have all the ingredients needed to find a multi-agent epistemic paradox in box world. 

\section{Finding the paradox}
\label{sec:paradox}

In this section we find a scenario in box world where reasoning, physical agents reach a logical paradox. We compare it to the result to the contradiction obtained by Frauchiger and Renner \cite{Frauchiger2018} in the next section.

\paragraph{Experimental setup.} The proposed thought experiment is similar in spirit to the one proposed by Frauchiger and Renner~\cite{Frauchiger2018} (recall Figure~\ref{fig:circuits}). Alice and Bob share a PR box (the corresponding box world state is given in Figure~\ref{fig:PRbox}); they each will measure their half of the PR box and store the outcomes in their local memories. Let Alice's lab be located inside the lab of another agent, Ursula's lab such that Ursula can now perform joint measurements on Alice's system (her half of the PR box) and memory, as seen in the previous section. Similarly, let Bob's lab be located inside Wigner's lab, such that Wigner can perform joint measurements on Bob's system and memory. We assume that Alice's and Bob's labs are isolated such that no information about their measurement outcomes leaks out. 
The protocol, as shown in Figure~\ref{fig:circuitPRbox}, is the following:
\begin{itemize}
    \item[\bf t=1] Alice measures her half of the PR box, with measurement setting $X$, and stores the outcome $a$ in her memory $A$. 
    \item[\bf t=2] Bob measures his half of the PR box, with measurement setting $Y$, and stores the outcome $b$ in his memory $B$. 
    \item[\bf t=3] Ursula measures the box corresponding to Alice's lab (as in Figure~\ref{fig:update_ursula}), with measurement setting $\tilde X = X\oplus 1$, obtaining outcome $\tilde a$. 
    \item[\bf t=4] Wigner measures the box corresponding to Bob's lab, with measurement setting $\tilde Y = Y\oplus 1$, obtaining outcome $\tilde b$. 
\end{itemize}
Agents can agree on their measurement settings beforehand, but should not communicate once the experiment begins. The trust relation, which specifies which agents consider each other to be rational agents (as opposed to mere physical systems), is 
\begin{align*}
    A_{t=1,2} &\leftrightsquigarrow B_{t=1,2} \\ B_{t=2,3} &\leftrightsquigarrow U_{t=3} \\ U_{t=3,4} &\leftrightsquigarrow W_{t=4}\\  W_{t=4} &\leftrightsquigarrow A_{t=1}.
\end{align*}
The common knowledge $\mathbbm T$ shared by all four agents includes the PR box correlations, the way the external agents model Alice and Bob's measurements, and the trust structure above.  

\paragraph{Reasoning.} Now the agents can reason about the events in other agents' labs.
We take here the example where the measurement settings are $X=Y=0, \tilde X = \tilde Y=1$, and where Wigner obtained the outcome $\tilde b=0$; the reasoning is analogous for the remaining cases. 
\begin{enumerate}
    \item \emph{Wigner reasons about Ursula's outcome.} At time $t=4$, Wigner knows that, by virtue of their information-preserving modelling of Alice and Bob's measurements, he and Ursula effectively shared a PR box \footnoteremember{proof}{See Appendix~\ref{appendix:proofs} for a proof.}. 
    He can therefore use the PR correlations $\tilde X \tilde Y = \tilde a \oplus \tilde b$ to conclude that Ursula's output must be $1$, 
    $$ K_W (\tilde{b}=0 \implies \tilde a=1).$$
    
    \item \emph{Wigner reasons about Ursula's reasoning.} Now Wigner thinks about what Ursula may have concluded regarding Bob's outcome. He knows that at time $t=3$, Ursula and Bob effectively shared a PR box\footnoterecall{proof}, satisfying $\tilde X Y = \tilde a \oplus b$, and that therefore Ursula must have concluded
    $$ K_W K_U (\tilde{a}=1 \implies b=1).$$
    
    \item \emph{Wigner reasons about Ursula's reasoning about Bob's reasoning.} Next, Wigner wonders ``What could Ursula, at time $t=3$, conclude about Bob's reasoning at time $t=2$?" Well, Wigner knows that she knows that Bob knew that at time $t=2$ he effectively shared a PR box with Alice, satisfying $XY=a\oplus b$, and therefore concludes
     $$ K_W K_U K_B (b=1 \implies a=1).$$
     
    \item \emph{Wigner reasons about Ursula's reasoning about Bob's reasoning about Alice's reasoning.} We are almost there. Now Wigner thinks about Alice's perspective at time $t=1$, through the lenses of Bob (at time $t=2$) and Ursula ($t=3$). Back then, Alice knew that she obtained some outcome $a$, and that Wigner would model Bob's measurement in an information-preserving way, such that Alice (at time $t=1$) and Wigner (of time $t=4$) share an effective PR box\footnoterecall{proof}, satisfying $X \tilde Y = a \oplus \tilde b$, which results, in particular, in
    $$ K_W K_U K_B K_A (a=1 \implies \tilde b=1).$$
    
    \item \emph{Wigner applies trust relations.} In order to combine the statements obtained above, we need to apply the trust relations described above, starting from the inside of each proposition, for example, 
    \begin{align*}
         K_W K_U K_B K_A &(a=1 \implies \tilde b=1) \\
         \implies K_W K_U K_B  &(a=1 \implies \tilde b=1) \qquad \lidia{A \leadsto B}\\
         \implies K_W K_U  &(a=1 \implies \tilde b=1) \qquad \lidia{B \leadsto U}\\
         \implies K_W &(a=1 \implies \tilde b=1), \qquad \lidia{U \leadsto W}
    \end{align*}
    and similarly for the other statements, so that we obtain 
    \begin{align*}
         & K_W  \big[ (\tilde{b}=0 \implies \tilde a=1) \wedge (\tilde{a}=1 \implies b=1) \wedge  (b=1 \implies a=1) \wedge (a=1 \implies \tilde b=1) \big] \\
    \implies &K_W (\tilde{b}=0 \implies \tilde b =1).
    \end{align*}
        
\end{enumerate}
We could have equally taken the point of view of any other observer, and from any particular outcome or choice of measurement, and through similar reasoning chains reached the following contradictions, 
\begin{align*}
    &K_A [ ({a}=0 \implies a =1) \wedge ({a}=1 \implies  a =0) ],\\
    &K_B [ (b=0 \implies b =1) \wedge (b=1 \implies b =0) ],\\
    &K_U [ (\tilde{a}=0 \implies \tilde a =1) \wedge (\tilde{a}=1 \implies \tilde a =0) ],\\
    &K_W [ (\tilde{b}=0 \implies \tilde b =1) \wedge (\tilde{b}=1 \implies \tilde b =0) ].
\end{align*}

\section{Discussion}
\label{sec:conclusions}
We have generalized the conditions of the Frauchiger-Renner theorem and made them applicable to arbitrary physical theories, including the framework of \emph{generalized probability theories}. We then applied these conditions to the GTP of box world and found an experimental setting that leads to a multi-agent epistemic paradox.

\subsection{Comparison with the quantum thought experiment}

We showed that box world agents reasoning about each others' knowledge can come to a deterministic contradiction, which is stronger than the original paradox, as it can be reached without post-selection, from the point of view of every agent and for any measurement outcome obtained by them.

\paragraph{Strong contextuality and post-selection.}
In contrast to the original Frauchiger-Renner experiment of \cite{Frauchiger2018}, no post-selection was required to arrive at this contradictory chain of statements as, in fact, all the implications above are symmetric, for example
$$\tilde{a}=0 \iff b=0 \iff a=0 \iff \tilde{b}=0 \iff \tilde{a}=1.$$
As a result, one can arrive at a similar (symmetric) paradoxical chain of statements irrespective of the choice of agent and outcome for the first statement. In other words, irrespective of the outcomes observed by every agent, each agent will arrive at a contradiction when they try to reason about the outcomes of other agents. This is because, as shown in \cite{Abramsky15}, the PR box exhibits strong contextuality and no global assignments of outcome values for all four measurements exists for any choice of local assignments. In contrast, the original paradox of \cite{Frauchiger2018} admits the same distribution as that of Hardy's paradox \cite{Hardy1993}. It is shown in \cite{Abramsky15} that this distribution is an example of logical contextuality where for a particular choice of local assignments (the ones that are post-selected on in the original Frauchiger-Renner experiment), a global assignment of values compatible with the support of the distribution fails to exist, but this is not true for all local assignments. This makes the paradox even stronger in box world, since it can be found without post-selection and by any of the agents, for any outcome that they observe. In particular, the paradox would already arise in a single run of the experiment. For a simple method to enumerate all possible contradictory statements that the agents may make, see the analysis of the PR box presented in \cite{Abramsky15}. A detailed analysis of the relation between Frauchiger-Renner type paradoxes and contextuality will appear in future work.

\paragraph{Communication \emph{vs} prepare-and-measure.}
One might note that a distinction between our proposal and  the original Frauchiger-Renner experiment is that there is no communication between Alice and Bob in our PR box version. However, the original quantum  scenario can be replaced by a protocol where Alice and Bob receive an appropriately prepared quantum state and perform measurements on it without communicating to each other, and the original paradox would still hold in such a case (Figure~\ref{fig:circuitFR}).

\subsection{Physical measurements in box world}
\label{sec:discussion}

Since we lack a physical theory to explain how measurements and transformations are instantiated for generalised non-signalling boxes, and only have access to their input/output behaviour, all allowed transformations consist of pre- and post-processing. In the quantum case, we have in addition to a description of possible input-output correlations, a mathematical framework for the underlying states producing those correlations, the theory of von-Neumann measurements and transformations as CPTP maps. In Appendix~\ref{appendix:fr} we briefly show how we one could in principle model the quantum memory updates in the framework of GPTs. 
In box world, introduction of dynamical features (for example, a memory update algorithm)  is less intuitive and requires additional constructions. In the following, we outline the main limitations we found.

\paragraph{Systems vs boxes.} In quantum theory, a system corresponds to a physical substrate that can be acted on more than once. For example, Alice could measure a spin first in the $Z$ basis and then in $X$ basis (obviously with different results than if she had measured first $X$ and then $Z$). The predictions for each subsequent measurement are represented by a different box in the GPT formalism, such that each box encodes the current state of the system in terms of the measurement statistics of a tomographically complete set of measurements. After each measurement, the corresponding box disappears, but quantum mechanics gives us a rule to compute the post-measurement state of the underlying system, which in turn specifies the box for future measurements. On the other hand, the default theory for box world lacks the notion of underlying physical systems and a definite rule to compute the post-measurement vector state of something that has been measured once.  Indeed, Equations~\ref{eq: op1}-\ref{eq: op3} (Appendix~\ref{appendix:boxes}) tell us that post-measurement states is only partially specified: for instance, if  the measurement performed was fiducial, we know that the block corresponding to that measurement in the post-measurement state would have a ``1'' corresponding to the outcome obtained and ``0'' for all other outcomes in the block. However, we still have freedom in defining the entries in the remaining blocks.
Our model proposes a possible physical mechanism for updating boxes (which could be read as updating the state of the underlying system), but so far only for the case where we compare the perspectives of different agents, and we leave it open whether Alice has a subjective update rule that would allow her to make subsequent measurements on the same physical system.

\paragraph{Verifying a measurement.} In our simple model, the external observer Ursula has no way to know which measurement Alice performed, or whether she measured anything at all --- the connection between Alice's and Ursula's views is postulated rather than derived from a physical theory. Indeed,  Alice could have simply wired the boxes as in Figure \ref{fig:update_copy} without actually performing the measurement, and Ursula will not know the difference: she obtains the same joint state of Alice's memory and the system she measured.
In contrast, consider the case of quantum mechanics with standard von Neumann measurements. There, Alice's memory gets entangled with the system, and the post-measurement state depends on the basis in which Alice measured her system. For example, if Alice's qubit $S$ starts off in the normalised pure state $\ket{\psi}=\alpha\ket{0}_S+\beta\ket{1}_S$ and her memory $M$ initialised to $\ket{0}_M$, the initial state of her system and memory from Ursula's perspective is $\ket{\Psi}^{in}_{SM}= [\alpha\ket{0}_S+\beta\ket{1}_S]\otimes \ket{0}_M=[(\frac{\alpha+\beta}{\sqrt{2}})\ket{+}_S+(\frac{\alpha-\beta}{\sqrt{2}})\ket{-}_S)]\otimes \ket{0}_M$
If Alice measures the system in the $Z$ basis, the post-measurement state from Ursula's perspective is $ \ket{\Psi}^{out,Z}_{SM}= \alpha\ket{0}_S\ket{0}_M+\beta\ket{1}_S\ket{1}_M$, which is an entangled state. If instead, Alice measured in the Hadamard (X) basis, the post-measurement state would be $ \ket{\Psi}^{out,X}_{SM}
   =(\frac{\alpha+\beta}{\sqrt{2}})\ket{+}_S\ket{0}_M+(\frac{\alpha-\beta}{\sqrt{2}})\ket{-}_S\ket{1}_M$. Clearly the measurement statistics of $\ket{\Psi}^{in}_{SM}$, $ \ket{\Psi}^{out,Z}_{SM}$ and  $\ket{\Psi}^{out,X}_{SM}$ are different and Ursula can thus (in principle, with some probability) tell whether or not Alice performed a measurement and which measurement was performed by her. 
In the absence of a physical theory backing box world, we can still lift this degenerancy between the three situations (Alice didn't measure, she measured $X=0$, or she measured $X=1$) by adding another classical system to the circuitry of~\ref{fig:update_copy}: for example, a trit that stores what Alice did, and which Ursula could consult independently of the glued box of system and Alice's memory. However, we'd still have a postulated connection between what's stored in this trit and what Alice actually did, and not one that is physically motivated.

\paragraph{\emph{Supergluing} of non-signalling boxes.}
For the memory update circuit (from Ursula's perspective) of Figure \ref{fig:update_copy}, and the initial state of Equation~\ref{eq: PSMin}, the final state would be $\vec{P}_{fin}^{SM}=(p\quad 0\quad 0\quad 1-p|p\quad 0\quad 0\quad 1-p|q\quad 0\quad 0\quad 1-q|q\quad 0\quad 0\quad 1-q)^T_{SM}$. Note that while the reduced final state of $S$ does not depend on the input $X'$ to $M$, the reduced final state on Alice's memory $M$, $\vec{P}_{fin}^{M}$ clearly depends on the input $X$ of the system $S$ if $p\neq q$. If $X=0$, $\vec{P}^M_{fin}=(p\quad 1-p|p\quad 1-p)^T$ and if $X=1$, $\vec{P}^M_{fin}=(q\quad 1-q|q\quad 1-q)^T$, i.e., the systems $S$ and $M$ are \emph{signalling}. This is expected since there is clearly a transfer of information from $S$ to $M$ during the measurement as seen in Figure~\ref{fig: bipartitemeas}. However, this means that the state $\vec{P}_{fin}^{SM}$ is not a valid box world state of 2 systems $S$ and $M$ but a valid state of a single system $SM$ i.e., after Alice performs her wiring/measurement, it is not possible to physically separate Alice's system $S$ from her memory $M$ from Ursula's perspective. For if this were possible, there would be a violation of the no-signalling principle and the notion of relativistic causality. In quantum theory, on the other hand it is always possible to perform separate measurements on Alice's system and her memory even after she measures. We call this feature \emph{supergluing} of post-measurement boxes, where it is no longer possible for Ursula to separately measure $S$ or $M$, but she can only jointly measure $SM$ as though it were a single system. Note that this is only the case for $p\neq q$ and in our example with the PR-box (Section \ref{sec:paradox}), $p=q=1/2$ and $\vec{P}_{fin}^{SM}$ remains a valid bipartite non-signalling state in this particular, fine-tuned case of the PR box.

\paragraph{A glass half full.}
The above-mentioned  features of the memory update in box world are certainly not desirable, and not what one would expect to find in a physical theory with meaningful notions of subsystems. An optimistic way to look at these limitations  is to see them as providing us with further intuition for why PR boxes have not yet been found in nature. One of the main contributions of this paper is the finding that despite these peculiar features of box world and the fact that it has no entangling bipartite joint measurements (a crucial step in the original quantum paradox), a consistent outside perspective of the memory update exists such that with our generalised assumptions, a multi-agent paradox can be recovered. This indicates that the  reversibility of measurements is not crucial to derive this kind of paradox.

\paragraph{Other models for physical measurements.}
Ours is not the first attempt at coming up with a (partial) physical theory that reproduces the statistics of box world. Here we review the  approach of Skrzypczyk et al.\ in~\cite{Skrzypczyk2008}. There the authors consider a variation of box world that has a reduced set of physical states (which the authors call \emph{genuine}), which consists of the PR box and all the deterministic local boxes. The wealth of box world state vectors (i.e.\ the non-signalling polytope, or what we could call  epistemic states) is recovered by allowing classical processing of inputs and outputs via classical wirings, as well as convex combinations thereof. In contrast, box world takes all convex combinations of maximally non-signalling boxes (of which the PR box is an example) to be genuine physical states; this becomes relevant as we require the allowed physical operations to map such states to each other. For the restricted state space of~\cite{Skrzypczyk2008}, the set of allowed combinations is larger than in box world, particularly for multipartite settings. For example, there we are allowed maps that implement the equivalent of entanglement swapping: if Bob shares a PR box with Alice, and another with Charlie, there is an allowed map that he can apply on his two halves which leaves Alice and Charlie sharing a PR box, with some probability. It would be interesting to try to model memory update in this modified theory, to see if (1) there is a more natural implementation of measurements within the extended set of operations, and (2) whether this theory allows for multi-agent paradoxes.

\subsection{Characterization of general theories}

While we have shown that a consistency paradox, similar to the one arising in the Frauchiger-Renner setup, can also be adapted for the box world in terms of GPTs, it still remains unclear how to characterize all possible theories where it is possible to find a setup leading to a contradiction. Essentially, one has to restrict the class of such theories and identify the properties of these theories that make such paradoxes possible. 

\paragraph{Beyond standard composition of systems.} 
Additionally, it is still an open problem to find an operational way to state the outside view of measurements (and a memory update operation), for theories without a prior notion of subsystems and a tensor rule for composing them. This will allow us to search for multi-agent logical paradoxes in field theories, for example. One possible direction is to use  notions of effective and subjective locality, as outlined for example in~\cite{Kraemer2018}.

\paragraph{Relation to contextuality.}
In~\cite{Abramsky15} Abramsky et al explore relations between logical paradoxes and quantum contextuality; in particular, they point out a direct connection between contextuality and a type of classic semantic paradoxes called ``Liar cycles''\cite{Cook2004}. A \textit{Liar cycle} of length N is a chain of statements of the  form:
\begin{align}
    \phi_1 = ``\phi_2 \ \text{is true}'', \phi_2 = ``\phi_3 \ \text{is true}'', \dots, \phi_{N-1} = ``\phi_N \ \text{is true}'', \phi_N = ``\phi_1 \ \text{is false}''.
\end{align}
It can be shown that the patterns of reasoning which are used in finding a contradiction in the chain of statements above are similar to the reasoning we make use of in FR-type arguments, and can also be connected to the cases of PR box (which corresponds to a Liar cycle of length 4) or Hardy's paradox. This might imply that multi-agent paradoxes are linked to the notion of contextuality. Another central ingredient seems to be information-preserving models for physical measurements, which allow us to replace counter-factuals with actual measurements, performed in sequence by different agents.
We leave a deeper investigation of these connections further to future work.

\newpage
\begin{acknowledgements}
We thank Roger Colbeck, Matt Leifer, Sandu Popescu and Renato Renner for valuable discussions, and Ravi Kunjwal for the working title of this paper, \emph{PRdoxes}.
VV acknowledges support from the Department of Mathematics, University of York.
NN and LdR acknowledge support from the Swiss
National Science Foundation through 
SNSF project No.\ $200020\_165843$ and through the 
the National Centre of
Competence in Research \emph{Quantum Science and Technology}
(QSIT). 
LdR further acknowledges support from  the FQXi grant \emph{Physics of the observer}. 
\end{acknowledgements}

%............. Appendix ..............
%..... Bibliography .....

%\bibliographystyle{unsrtnat}

%\bibliography{paradoxes}

\begin{thebibliography}{21}
\providecommand{\natexlab}[1]{#1}
\providecommand{\url}[1]{\texttt{#1}}
\expandafter\ifx\csname urlstyle\endcsname\relax
  \providecommand{\doi}[1]{doi: #1}\else
  \providecommand{\doi}{doi: \begingroup \urlstyle{rm}\Url}\fi

\bibitem[Frauchiger and Renner(2018)]{Frauchiger2018}
Daniela Frauchiger and Renato Renner.
\newblock Quantum theory cannot consistently describe the use of itself.
\newblock \emph{Nature Communications}, 9\penalty0 (1):\penalty0 3711, 2018.
\newblock ISSN 2041-1723.
\newblock \doi{10.1038/s41467-018-05739-8}.

\bibitem[Hardy(2001)]{Hardy01}
Lucien Hardy.
\newblock Quantum theory from five reasonable axioms, 2001.
\newblock
  \href{https://arxiv.org/abs/quant-ph/0101012}{arXiv:quant-ph/0101012}.

\bibitem[Barrett(2007)]{Barrett07}
Jonathan Barrett.
\newblock Information processing in generalized probabilistic theories.
\newblock \emph{Phys. Rev. A}, 75:\penalty0 032304, Mar 2007.
\newblock \doi{10.1103/PhysRevA.75.032304}.

\bibitem[Popescu and Rohrlich(1994)]{Popescu1994}
Sandu Popescu and Daniel Rohrlich.
\newblock Quantum nonlocality as an axiom.
\newblock \emph{Foundations of Physics}, 24\penalty0 (3):\penalty0 379--385,
  Mar 1994.
\newblock ISSN 1572-9516.
\newblock \doi{10.1007/BF02058098}.

\bibitem[Gross et~al.(2010)Gross, M\"uller, Colbeck, and Dahlsten]{Gross2010}
David Gross, Markus M\"uller, Roger Colbeck, and Oscar C.~O. Dahlsten.
\newblock All reversible dynamics in maximally nonlocal theories are trivial.
\newblock \emph{Phys. Rev. Lett.}, 104:\penalty0 080402, Feb 2010.
\newblock \doi{10.1103/PhysRevLett.104.080402}.

\bibitem[Nurgalieva and del Rio(2019)]{NL2018}
Nuriya Nurgalieva and Lídia del Rio.
\newblock Inadequacy of modal logic in quantum settings.
\newblock \emph{EPCTS}, 287:\penalty0 267--297, 2019.
\newblock \doi{10.4204/EPTCS.287.16}.

\bibitem[Von~Neumann(1955)]{vonNeumann1955}
John Von~Neumann.
\newblock \emph{Mathematical foundations of quantum mechanics}.
\newblock Number~2. Princeton university press, 1955.
\newblock ISBN 9780691178561.

\bibitem[Hardy(1993)]{Hardy1993}
Lucien Hardy.
\newblock Nonlocality for two particles without inequalities for almost all
  entangled states.
\newblock \emph{Phys. Rev. Lett.}, 71:\penalty0 1665--1668, Sep 1993.
\newblock \doi{10.1103/PhysRevLett.71.1665}.

\bibitem[Abramsky et~al.(2015)Abramsky, Barbosa, Kishida, Lal, and
  Mansfield]{Abramsky15}
Samson Abramsky, Rui~Soares Barbosa, Kohei Kishida, Raymond Lal, and Shane
  Mansfield.
\newblock {Contextuality, Cohomology and Paradox}.
\newblock In Stephan Kreutzer, editor, \emph{24th EACSL Annual Conference on
  Computer Science Logic (CSL 2015)}, volume~41 of \emph{Leibniz International
  Proceedings in Informatics (LIPIcs)}, pages 211--228, Dagstuhl, Germany,
  2015. Schloss Dagstuhl--Leibniz-Zentrum fuer Informatik.
\newblock ISBN 978-3-939897-90-3.
\newblock \doi{10.4230/LIPIcs.CSL.2015.211}.

\bibitem[Spekkens(2007)]{Spekkens07}
Robert~W. Spekkens.
\newblock Evidence for the epistemic view of quantum states: A toy theory.
\newblock \emph{Phys. Rev. A}, 75:\penalty0 032110, Mar 2007.
\newblock \doi{10.1103/PhysRevA.75.032110}.

\bibitem[Elga(2000)]{Elga2000}
Adam Elga.
\newblock Self-locating belief and the sleeping beauty problem.
\newblock \emph{Analysis}, 60\penalty0 (2):\penalty0 143--147, 2000.
\newblock \doi{10.1093/analys/60.2.143}.

\bibitem[Aaronson(2011)]{Aaronson2017}
Scott Aaronson.
\newblock Why philosophers should care about computational complexity.
\newblock \emph{CoRR}, abs/1108.1791, 2011.
\newblock URL \url{http://arxiv.org/abs/1108.1791}.

\bibitem[del Rio et~al.(2015)del Rio, Kr{\"a}mer, and Renner]{delRio2015}
L{\'\i}dia del Rio, Lea Kr{\"a}mer, and Renato Renner.
\newblock Resource theories of knowledge.
\newblock 2015.
\newblock \href{https://arxiv.org/abs/1511.08818}{arXiv:1511.08818}.

\bibitem[Kr{\"a}mer and del Rio(2018)]{Kraemer2018}
Lea Kr{\"a}mer and L{\'\i}dia del Rio.
\newblock Operational locality in global theories.
\newblock \emph{Philosophical Transactions of the Royal Society of London A:
  Mathematical, Physical and Engineering Sciences}, 376\penalty0 (2123), 2018.
\newblock ISSN 1364-503X.
\newblock \doi{10.1098/rsta.2017.0321}.

\bibitem[Grunhaus et~al.(1996)Grunhaus, Popescu, and Rohrlich]{Grunhaus1996}
Jacob Grunhaus, Sandu Popescu, and Daniel Rohrlich.
\newblock Jamming nonlocal quantum correlations.
\newblock \emph{Phys. Rev. A}, 53:\penalty0 3781--3784, Jun 1996.
\newblock \doi{10.1103/PhysRevA.53.3781}.

\bibitem[Horodecki and Ramanathan(2016)]{Horodecki16}
Pawe\l{} Horodecki and Ravishankar Ramanathan.
\newblock Relativistic causality vs. no-signaling as the limiting paradigm for
  correlations in physical theories, 2016.
\newblock \href{https://arxiv.org/abs/1611.06781}{arXiv:1611.06781}.

\bibitem[Short et~al.(2006)Short, Popescu, and Gisin]{Short2006}
Anthony~J Short, Sandu Popescu, and Nicolas Gisin.
\newblock Entanglement swapping for generalized nonlocal correlations.
\newblock \emph{Physical Review A}, 73\penalty0 (1):\penalty0 012101, 2006.
\newblock \doi{10.1103/PhysRevA.73.012101}.

\bibitem[Skrzypczyk et~al.(2008)Skrzypczyk, Brunner, and
  Popescu]{Skrzypczyk2008}
Paul Skrzypczyk, Nicolas Brunner, and Sandu Popescu.
\newblock Emergence of quantum correlations from non-locality swapping.
\newblock 2008.
\newblock \doi{10.1103/PhysRevLett.102.110402}.

\bibitem[Cook(2004)]{Cook2004}
Roy~T Cook.
\newblock Patterns of paradox.
\newblock \emph{The Journal of Symbolic Logic}, 69\penalty0 (3):\penalty0
  767--774, 2004.
\newblock \doi{10.2178/jsl/1096901765}.

\bibitem[Kripke(2012)]{Kripke2007}
Saul~A. Kripke.
\newblock Semantical considerations on modal logic.
\newblock In \emph{Universal Logic: An Anthology}, pages 197--208. Springer
  Basel, 2012.
\newblock \doi{10.1007/978-3-0346-0145-0_16}.

\bibitem[Garson(2016)]{LogicStanford}
James Garson.
\newblock Modal logic.
\newblock In Edward~N. Zalta, editor, \emph{The Stanford Encyclopedia of
  Philosophy}. Metaphysics Research Lab, Stanford University, spring 2016
  edition, 2016.
\newblock URL
  \url{https://plato.stanford.edu/archives/spr2016/entries/logic-modal/}.

\end{thebibliography}

 %for arxiv; compile again before uploading

\appendix

\addcontentsline{toc}{section}{\sc{Appendix}}

\section*{\textsc{Appendix}}

\section{Modal logic}
\label{appendix:logic}

Here we shortly sum up the important features of modal logic. Importantly, modal logic applies to most classical multi-agent setups, and simply provides a compact mathematical way to capture the intuitive laws commonly used for reasoning.

\subsection{Kripke structures}

In modal logic, a set $\Sigma$ of possible states  (or alternatives, or \emph{worlds}) is introduced \cite{Kripke2007}: for example, in a world $s_1$ the key value is $k=1$ and Eve does not know it, and in a state  $s_2$ Eve could know that $k=0$. The truth value of a proposition $\phi$ is then assigned depending on the possible world in $\Sigma$, and can differ from one possible world to another.
In order to formalize the simple rules agents use for reasoning, we will first provide a structure which serves as a complete picture of the setup the agents are in, and then discuss the elements of the structure. 
\begin{definition}{\textbf{(Kripke structure)}}
A Kripke structure $M$ for $n$ agents over a set of statements $\Phi$ is a tuple $\langle \Sigma, \pi,\mathcal{K}_1,...,\mathcal{K}_n\rangle$ where $\Sigma$ is a non-empty set of states, or possible worlds, $\pi$ is an interpretation, and $\mathcal{K}_i$ is a binary relation on $\Sigma$.

The interpretation $\pi$ is a map $\pi:\Sigma\times\Phi\to\{\textbf{true},\textbf{false}\}$, which defines a truth value of a statement $\phi\in\Phi$ in a possible world $s\in\Sigma$.

$\mathcal{K}_i$ is a binary equivalence relation on a set of states $\Sigma$, where $(s,t)\in\mathcal{K}_i$ if agent $i$ considers world $t$ possible given his information in the world $s$. 
\end{definition}

The truth assignment  tells  us if the proposition $\phi\in\Phi$ is true or false in a possible world $s\in\Sigma$; for example, if $\phi=$ ``Alice has a secret key,'' and $s$ is a world where there is an individual named Alice who indeed possesses a secret key, then $\pi(s,\phi)=\textbf{true}$. The truth value of a statement in a given structure $M$ might vary from one possible world to another; we will denote that $\phi$ is true in world $s$ of a structure $M$ by $(M,s)\models\phi$, and $\models\phi$ will mean that $\phi$ is true in any world $s$ of a structure $M$.

\subsection{Axioms of knowledge (weak version)}

In order to operate the statements agents produce, we have to establish certain rules which are used to compress or judge the statements. These are the axioms of knowledge \cite{LogicStanford}. They might seem trivial in the light of our everyday reasoning, yet given our awareness of the quantum case, we will treat them carefully. Here we present the reader with a weaker version of the axioms (which includes Trust axiom) that we have developed in previous work \cite{NL2018}.

Distribution axiom allows agents combine statement which contain inferences:
\begin{axiom}[Distribution axiom.]
If an agent is aware of a fact $\phi$ and that a fact $\psi$ follows from $\phi$, then the agent can conclude that $\psi$ holds:
$$(M,s)\models(K_i\phi\wedge K_i(\phi\Rightarrow\psi))\Rightarrow (M,s)\models K_i\psi.$$
\end{axiom}

Knowledge generalization rule permits agents use commonly shared knowledge:
\begin{axiom}[Knowledge generalization rule.]
All agents know all the propositions that are  valid in a structure:
\begin{center}
if $(M,s)\models\phi \ \forall s$ then $\models K_i\phi \ \forall i.$
\end{center}
\end{axiom}

Positive and negative introspection axioms highlight the ability of an agent to reflect upon her knowledge:
\begin{axiom}[Positive and negative introspection axioms.] 
Agents can perform  introspection regarding their knowledge:
\begin{center}
$(M,s)\models K_i\phi\Rightarrow (M,s)\models K_iK_i\phi$ (Positive Introspection),\\
$(M,s)\models \neg K_i\phi\Rightarrow (M,s)\models K_i\neg K_i\phi$ (Negative Introspection).
\end{center}
\end{axiom}

We also equip the logical skeleton of the setting with so-called trust structure, which governs the way the information is passed on between agents:
\begin{definition}[Trust]
\label{def:trust}
    We say that an agent $i$ trusts an agent $j$ (and denote it by $j \leadsto i$ ) if and only if $$K_i \ K_j \ \phi \implies K_i \ \phi,$$ for all $\phi$. 
\end{definition}

In the Frauchiger-Renner setup, as well as in the thought experiment presented in this paper, we consider the following trust structure between agents:
\begin{align}
    A \leadsto B \leadsto U\leadsto  W \leadsto A.
    \label{eq:trustFR}
\end{align}

Further discussion on axioms of modal logic and their application in quantum mechanics can be found in our paper \cite{NL2018}.

\section{Generalized probabilistic theories}
\label{appendix:boxes}

In quantum theory, systems are described by states that live in a Hilbert space, measurements and transformations on these states are represented by CPTP maps and the Born rule specifies how to obtain the probabilities of possible measurement outcomes give these states and measurements. In more general theories, there is no reason to assume Hilbert spaces or CPTP maps. In fact such a description of the state space and operations may not even be available, systems may be described as black boxes taking in classical inputs (choice of measurements) and giving classical outputs (measurement outcomes). What we can demand is that the theory provides a way for agents to predict the probabilities of obtaining various outputs based on their input choice and some operational description of the box.

Barrett derived the mathematical structure of the state-space of composite systems and allowed operations on systems from a few reasonable, physically motivated assumptions \cite{Barrett07}. We follow his formalism here.  Later, Gross \emph{et al.}\ found restrictions on the reversible dynamics of maximally non-local GPTs \cite{Gross2010} showing that all reversible operations on box-world are trivial i.e., they map product states to product states and cannot correlate initially uncorrelated systems. In accordance with this, our memory update procedure that maps the initial product state $\vec{P}_{in}^{SM}$ (Equation~\ref{eq: PSMin}) to the final correlated state of the system and memory $ \vec{P}_{fin}^{SM}$ (Equation~\ref{eq: PSMfin} or equivalently Equation~\ref{eq: finshort}) is an irreversible transformation in contrast to the quantum case where the corresponding transformation is a unitary and hence reversible.

 \subsection{Observing outcomes}
 \label{ssec: boxworldoutcomes}

In Section~\ref{sec:boxworld}, we briefly reviewed states and transformations in GPTs, in particular box world; here we go into further detail. Consider a GPT, $\mathbb{T}$. Denoting the set of all allowed states of a system in $\mathbb{T}$ by $\mathcal{S}$, any valid transformation on a normalised GPT state $\vec{P}\in \mathcal{S}$ maps it to another normalised GPT state in $\mathcal{S}$. Consequently, is linear and can be represented by a matrix $M$ such that $\vec{P}\rightarrow M.\vec{P}$ under this transformation and $M.\vec{P}\in \mathcal{S}$ \cite{Barrett07}. Further, operations that result in different possible outcomes can be associated with a set of transformations, one for each outcome. These also give an operational meaning to unnormalised states where $|\vec{P}|=\sum_i P(a=i|X=j)=c \quad \forall j, c\in [0,1]$ (i.e., the norm is independent of the value of $j$). Such an operation $M$ on a normalised initial state $\vec{P}$ can be associated with a set of matrices $\{M_i\}$ such that the unnormalised state corresponding to the $i^{th}$ outcome is $M_i.\vec{P}$. Then the probability of obtaining this outcome is simply the norm of this unnormalised state, $|M_i.\vec{P}|$ and the corresponding normalized final state is $M_i.\vec{P}/|M_i.\vec{P}|$. A set $\{M_i\}$ represents a valid operation if the following hold \cite{Barrett07}.
\begin{subequations}
\begin{equation}
\label{eq: op1}
    0 \leq |M_i.\vec{P}| \leq 1 \qquad \forall i, \vec{P} \in \mathcal{S}
\end{equation}
\begin{equation}
\label{eq: op2}
    \sum_i |M_i.\vec{P}|=1 \qquad \forall \vec{P}\in \mathcal{S}
\end{equation}
\begin{equation}
\label{eq: op3}
    M_i.\vec{P} \in \mathcal{S} \qquad \forall i, \vec{P}\in \mathcal{S}
\end{equation}
\end{subequations}

 This is the analogue of quantum Born rule for GPTs. Box world is a GPT where the state space $\mathcal{S}$ consists of all normalized states $\vec{P}$ whose entries are valid probabilities (i.e., $\in [0,1]$) and satisfy the \emph{no-signalling} constraints i.e., for a $N$-partite state $\vec{P}$, the marginal term $\sum_{a_i}P(a_1,..,a_i,..,a_N|X_1,..,X_i,..,X_N)$ is independent of the setting $X_i$ forall $i\in \{1,...,N\}$\footnote{This is in the spirit of relativistic causality since one would certainly expect that the input of one party does not affect the output of others when the are all space-like separated from each other.}
 
 When the GPT $\mathbb{T}$ is box world, the conditions of Equations~\ref{eq: op1}-\ref{eq: op3} result in the characterization of measurements and transformations in the theory in terms of classical circuits or \emph{wirings} as shown in \cite{Barrett07}. It suffices for the purpose of this paper to take that characterisation as the common knowledge of agents in the theory. In the original quantum paradox \cite{Frauchiger2018}, the Born rule is taken as common knowledge and here, the common knowledge consists of characterisations that follow from the box world analogue of the born rule (Equations~\ref{eq: op1}-\ref{eq: op3}).
We summarise the results of \cite{Barrett07} characterising allowed transformations and measurements in box world and will only consider normalization-preserving transformations.
\begin{itemize}
    \item \textbf{Transformations: } \begin{itemize}
        \item \emph{Single system:} All transformations on single box world systems are relabellings of fiducial measurements or outcomes or a convex combination thereof.
        \item \emph{Bipartite system:} Let $X$ and $Y$ be fiducial measurements performed on the transformed bipartite system with corresponding outcomes $a$ and $b$, then all transformations of 2-gbit systems can be decomposed into convex combinations of classical circuits of the following form: A fiducial measurement $X'=f_1(X,Y)$ is performed on the initial state of the first gbit resulting in the outcome $a'$ followed by a fiducial measurement $Y'=f_2(X,Y,X')$ on the initial state of the second gbit resulting in the outcome $b'$. The final outcomes are given as $(a,b)=f_3(X,Y,a',b')$, where $f_1$, $f_2$ and $f_3$ are arbitrary functions.
    \end{itemize}
     \item \textbf{Meaurements: } \begin{itemize}
        \item \emph{Single system:} All measurements on single box world systems are either fiducial measurements with outcomes relabelled or convex combinations of such.
        \item \emph{Bipartite system:} All bipartite measurements on 2-gbit systems can be decomposed into convex combinations of classical circuits of the following form (Figure~\ref{fig: bipartitemeas}): A fiducial measurement $X$ is performed on the initial state of the first gbit resulting in the outcome $a'$ followed by a fiducial measurement $Y=f(a')$ on the second gbit resulting in the outcome $b'$. The final outcome is $a=f'(a',b')$, where $f$ and $f'$ are arbitrary functions.
    \end{itemize}
\end{itemize}

 \begin{remark}
 Note that an agent Alice who measures a box world system only sees a classical final state, which corresponds the classical measurement outcome, since the box is a single-shot input/output function. Alice could use Equations~\ref{eq: op1}-\ref{eq: op3} to calculate the probabilities of obtaining different outcomes given the measurement she performs and prepare a new box (a new input/output function) depending on the measurement and outcome she just obtained (and has stored in her memory), as in Figure~\ref{fig:measurement_Alice}. An outside agent who does not know Alice's measurement outcome would see correlations between Alice's system and memory and would describe the measurement by an irreversible transformation, more specifically  a classical wiring between Alice's system and memory as shown in the following section.
 \end{remark}

\section{Memory update in box world (proofs)}
\label{appendix:proofs}

\subsection{Single lab}

In this section, we describe how a box world agent would measure a system and store the result in a memory. From the perspective of an outside observer (who does not know the outcome of the agent's measurement), we describe the initial and final states of the system and memory before and after the measurement as well as the transformation that implements this memory update in box world. In the quantum case, any initial state of the system $S$ is mapped to an isomorphic joint state of the system $S$ and memory $M$ (see Equation~\ref{eq:entangling_measurement}) and hence the memory update map that maps the former to the latter (an isometry in this case\footnote{An isometry since it introduces an initial pure state on $M$, followed by a joint unitary on $SM$.}) satisfies Definition~\ref{def:memory_update} of an information-preserving memory update. We will now characterise the analogous memory update map in box world and show that it also satisfies Definition~\ref{def:memory_update}.

\begin{theorem}
\label{theorem: boxmemory}
In box world, there exists a valid transformation $u$ that maps every arbitrary, normalized state $\vec{P}^S_{in}$ of the system $S$ to  an isomorphic final state  $\vec{P}_{fin}^{SM}$ of the system $S$ and memory $M$ and hence constitutes an information-preserving memory update (Definition~\ref{def:memory_update}).
\end{theorem}
\begin{proof}
To simplify the argument, we will describe the proof for the case where $S$ and $M$ are gbits. For higher dimensional systems, a similar argument holds, this will be explained at the end of the proof. 

We start with the system in an arbitrary, normalized gbit state $\vec{P}^S_{in}=(p\quad 1-p|q\quad 1-q)^T$ (where the subscript T denotes transpose and $p,q \in [0,1]$) and the memory initialised to one of the 4 pure states\footnote{It does not matter which pure state the memory is initialized in, a similar argument applies in all cases.}, say $\vec{P}^M_{in}=\vec{P}_1=(1\quad 0|1\quad 0)^T$. Then the joint initial state, $ \vec{P}_{in}^{SM}=(p\quad 1-p|q\quad 1-q)^T_S\otimes (1\quad 0|1\quad 0)^T_M$ of the system and memory can be written as follows, where $P_{in}(a=i,a'=j|X=k,X'=l)$ denotes the probability of obtaining the outcomes $a=i$ and $a'=j$ when performing the fiducial measurements $X=k$ and $X'=l$ on the system and memory respectively, in the initial state $\vec{P}_{in}^{SM}$.
\begin{equation}
\label{eq: PSMin}
    \vec{P}_{in}^{SM}=\left(
\begin{array}{c}
P_{in}(a=0,a'=0|X=0,X'=0)\\
P_{in}(a=0,a'=1|X=0,X'=0)\\
P_{in}(a=1,a'=0|X=0,X'=0)\\
P_{in}(a=1,a'=1|X=0,X'=0)\\
\hline
P_{in}(a=0,a'=0|X=0,X'=1)\\
P_{in}(a=0,a'=1|X=0,X'=1)\\
P_{in}(a=1,a'=0|X=0,X'=1)\\
P_{in}(a=1,a'=1|X=0,X'=1)\\
 \hline 
P_{in}(a=0,a'=0|X=1,X'=0)\\
P_{in}(a=0,a'=1|X=1,X'=0)\\
P_{in}(a=1,a'=0|X=1,X'=0)\\
P_{in}(a=1,a'=1|X=1,X'=0)\\
 \hline 
P_{in}(a=0,a'=0|X=1,X'=1)\\
P_{in}(a=0,a'=1|X=1,X'=1)\\
P_{in}(a=1,a'=0|X=1,X'=1)\\
P_{in}(a=1,a'=1|X=1,X'=1)\\
\end{array}
\right)_{SM}=
\left(
\begin{array}{c}
p\\
0\\
1-p\\
0\\
\hline
p\\
0\\
1-p\\
0\\
 \hline 
q\\
0\\
1-q\\
0\\
 \hline 
q\\
0\\
1-q\\
0\\
\end{array}
\right)_{SM}
\end{equation}

The rest of the proof proceeds as follows: we first describe a final state $\vec{P}_{fin}^{SM}$ of the system and memory and a corresponding memory update map $u$ that satisfy Definition~\ref{def:memory_update} of a generalized information-preserving memory update. Then, we show that this map is an allowed box world transformation which completes the proof. 

If an agent performs a measurement on the system, the state of the memory must be updated depending on the outcome and the final state of the system and memory after the measurement must hence be a correlated (i.e., a non-product) state. Although the full state space of the 2 gbit system $SM$ is characterised by the 4 fiducial measurements $(X,X')\in\{(0,0),(0,1),(1,0),(1,1)\}$, Definition~\ref{def:memory_update} allows us to restrict possible final states to a useful subspace of this state space that contain correlated states of a certain form. The definition requires that for every map $\mathcal{E}_S$ on the system before measurement, there exists a corresponding map $\mathcal{E}_{SM}$ on the system and memory after the measurement that is operationally identical. Thus it suffices if the joint final state $\vec{P}_{fin}^{SM}$ belongs to a subspace of the 2 gbit state space for which only 2 of the 4 fiducial measurements are relevant for characterising the state, namely any 2 fiducial measurements on $\vec{P^{SM}_{fin}}$ that are isomorphic to the 2 fiducial measurements on $\vec{P}^S_{in}$.
Note that by definition of fiducial measurements, the outcome probabilities of any measurement can be found given the outcome probabilities of all the fiducial measurements and without loss of generality, we will only consider the case where the agents perform fiducial measurements on their systems.

A natural isomorphism between fiducial measurements on $\vec{P}_{in}^S$ and those on $\vec{P}_{fin}^{SM}$ to consider here (in analogy with the quantum case) is: $X=i \Leftrightarrow (X,X')=(i,i) \quad, \forall i\in \{0,1\}$ i.e., only consider the cases where the fiducial measurements performed on $S$ and $M$ are the same.
Now, in order for the states to be isomorphic or operationally equivalent, one requires that performing the fiducial measurements $(X,X')=(i,i)$ on $\vec{P}_{fin}^{SM}$ should give the same outcome statistics as measuring $X=0$ on $\vec{P}_{in}^S$. This can be satisfied through an identical isomorphism on the outcomes: $a=i \Leftrightarrow (a,a')=(i,i) \quad, \forall i\in \{0,1\}$. Then the final state of the system and memory, $\vec{P}^{SM}_{fin}$ will be of the form
\begin{equation}
\label{eq: PSMfin}
    \vec{P}_{fin}^{SM}=\left(
\begin{array}{c}
P_{fin}(a=0,a'=0|X=0,X'=0)\\
P_{fin}(a=0,a'=1|X=0,X'=0)\\
P_{fin}(a=1,a'=0|X=0,X'=0)\\
P_{fin}(a=1,a'=1|X=0,X'=0)\\
\hline
P_{fin}(a=0,a'=0|X=0,X'=1)\\
P_{fin}(a=0,a'=1|X=0,X'=1)\\
P_{fin}(a=1,a'=0|X=0,X'=1)\\
P_{fin}(a=1,a'=1|X=0,X'=1)\\
 \hline 
P_{fin}(a=0,a'=0|X=1,X'=0)\\
P_{fin}(a=0,a'=1|X=1,X'=0)\\
P_{fin}(a=1,a'=0|X=1,X'=0)\\
P_{fin}(a=1,a'=1|X=1,X'=0)\\
 \hline 
P_{fin}(a=0,a'=0|X=1,X'=1)\\
P_{fin}(a=0,a'=1|X=1,X'=1)\\
P_{fin}(a=1,a'=0|X=1,X'=1)\\
P_{fin}(a=1,a'=1|X=1,X'=1)\\
\end{array}
\right)_{SM}=
\left(
\begin{array}{c}
p\\
0\\
0\\
1-p\\
\hline
*\\
*\\
*\\
*\\
 \hline 
*\\
*\\
*\\
*\\
 \hline 
q\\
0\\
0\\
1-q\\
\end{array}
\right)_{SM},
\end{equation}

where $*$ are arbitrary, normalised entries and where $P_{fin}(a=i,a'=j|X=k,X'=l)$ denotes the probability of obtaining the outcomes $a=i$ and $a'=j$ when performing the fiducial measurements $X=k$ and $X'=l$ on the system and memory respectively, in the final state $\vec{P}_{fin}^{SM}$. 
This final state can be compressed since the only relevant and non-zero probabilities in $\vec{P}_{fin}^{SM}$ occur when $X=X'$ and $a=a'$. We can then define new variables $\tilde{X}$ and $\tilde{a}$ such that $X=X'=i \Leftrightarrow \tilde{X}=i$ and $a=a'=j \Leftrightarrow \tilde{a}=j$ for $i,j \in \{0,1\}$ and $\vec{P}_{fin}^{SM}$ can equivalently be written as in Equation~\ref{eq: finshort} which is clearly of the same form as $P_{in}^S$.

\begin{equation}
\label{eq: finshort}
    \vec{P}_{fin}^{SM}\equiv \left(
\begin{array}{c}
 P(\tilde{a}=0|\tilde{X}=0)\\
 P(\tilde{a}=1|\tilde{X}=0)\\
\hline
 P(\tilde{a}=0|\tilde{X}=1)\\
 P(\tilde{a}=1|\tilde{X}=1)\\
\end{array}
\right)_{SM}=\left(
\begin{array}{c}
 p\\
 1-p\\
\hline
 q\\
 1-q\\
\end{array}
\right)_{SM}
\end{equation}

Hence the initial state of the system, $\vec{P}^S_{in}=(p\quad 1-p|q\quad 1-q)^T$ (which is an arbitrary gbit state) is isomorphic to the final state of the system and memory, $\vec{P}_{fin}^{SM}$ (as evident from Equation~\ref{eq: finshort}) with the same outcome probabilities for $X=0,1$ and $\tilde{X}=0,1$. This implies that for every transformation $\mathcal{E}_S$ on the former, there exists a transformation $\mathcal{E}_{SM}$ on the latter
such that for all outside agents $A_j$ and for all $p,q \in [0,1]$ (i.e., all possible input gbit states on the system), $ \quad K_j \phi[\mathcal{E}_S(\vec{P}^S_{in})] \Rightarrow K_j \phi[\mathcal{E}_{SM}\circ\vec{P}^SM_{fin}]$, where $\vec{P}^{fin}_{SM}=u(\vec{P}^{in}_{S})$. Thus any map $u$ that maps $\vec{P}^{in}_{S}$ to $\vec{P}^{fin}_{SM}$ satisfies Definition~\ref{def:memory_update}. 

We now find a valid box world transformation that maps the initial state $\vec{P}_{in}^{SM}$ (Equation~\ref{eq: PSMin}) to any final state of the form $\vec{P}_{fin}^{SM}$(Equation~\ref{eq: PSMfin}). This fully characterises the memory update map $u:\vec{P}^{in}_{S}\rightarrow \vec{P}^{fin}_{SM}$ since $\vec{P}_{in}^{SM}$ is obtained from $\vec{P}_{in}^{S}$ by simply tensoring a pure state $(1\quad 0| 1 \quad 0)^T_M$. 

Noting that all bipartite transformations in box world can be decomposed to a classical circuit of a certain form (see Appendix~\ref{ssec: boxworldoutcomes} or the original paper \cite{Barrett07} for details), In Figure~\ref{fig: memorycircuit}, we construct an explicit circuit of this form that converts $\vec{P}_{in}^{SM}$ to $\vec{P}_{fin}^{SM}$.By construction, we only need to consider the case of $X=X'$ since for $X\neq X'$, the entries of $\vec{P}_{fin}^{SM}$ can be arbitrary and are irrelevant to the argument. For $X\neq X'$, one can consider any such circuit description and it is easy to see that $\vec{P}_{in}^{SM}=(p\quad 1-p|q\quad 1-q)^T_S\otimes (1\quad 0|1\quad 0)^T_M$ is indeed transformed into $\vec{P}_{fin}^{SM}=(p\quad 0\quad 0\quad 1-p|*\quad*\quad*\quad*|*\quad*\quad*\quad*|q\quad 0\quad 0\quad 1-q)^T_{SM}$ through the transformation $\mathcal{T}$ defined by these sequence of steps. For example, if the circuit description for the $X\neq X'$ case is same as that for the $X=X'$ case, then the resultant memory update map is equivalent to the circuit of Figure~\ref{fig:update_copy} which corresponds to performing a fixed measurement $X'=0$ on the initial state of $M$ and a classical CNOT on the output wire of $M$ controlled by the output wire of $S$\footnote{The output wires of boxes carry classical information after the measurement.}. The final state in that case is $(p\quad 0\quad 0\quad 1-p|p\quad 0\quad 0\quad 1-p|q\quad 0\quad 0\quad 1-q|q\quad 0\quad 0\quad 1-q)^T_{SM}$.
Note that the memory update transformation $\mathcal{T}: \vec{P}_{in}^{SM} \rightarrow \vec{P}_{fin}^{SM}$ and hence the resulting map $u$ are not reversible. This is expected since the initial state $\vec{P}_{in}^{SM}$ is a product state while the final state $\vec{P}_{fin}^{SM}$  clearly is not (since $S$ and $M$ are correlated for an outside observer), and \cite{Gross2010} shows that all reversible transformations in box world map product states to product states.

\begin{figure}[H]
    \centering
\begin{tikzpicture}
\draw[blue!75!black,thick] (-0.5,-5) rectangle (9.1,2.5);
\draw (0,0) rectangle node[align=center]{$\mathbf{X_1=X}$} (2.4,1);
\draw (5,0) rectangle node[align=center]{$\mathbf{X_2=0}$} (7.4,1);
\draw[arrows={-stealth},blue!75!black,thick] (0.6,3.5)--(0.6,2.5); \draw[arrows={-stealth},blue!75!black,thick] (1.8,3.5)--(1.8,2.5);
\node[align=center,blue!75!black] at (0.6,3.8) {$\mathbf{X}$}; \node[align=center,blue!75!black] at (1.8,3.8) {$\mathbf{X'}$};
\draw[arrows={-stealth}] (0.6,2.5)--(0.6,1); \draw[arrows={-stealth}] (1.8,2.5)--(1.8,1);
\draw[arrows={-stealth}] (1.2,0)--(1.2,-1); \draw[arrows={-stealth}] (6.2,0)--(6.2,-1);
\node[align=center] at (1.5,-0.5) {$X_1$}; \node[align=center] at (6.5,-0.5) {$X_2$};
\draw (0.6,-2) rectangle node[align=center]{$\mathbf{S}$} (1.8,-1);
\draw (5.6,-2) rectangle node[align=center]{$\mathbf{M}$} (6.8,-1); \draw[dashed] (1.8,-1.5)--(5.6,-1.5); \draw[arrows={-stealth}] (6.2,-2)--(6.2,-3); \node[align=center] at (6.5,-2.5) {$a_2$}; \draw (5.6,-4.2) rectangle node[rectangle split,rectangle split parts=2]{$\mathbf{a=a_1}$ \nodepart{second} $\mathbf{a'=a_2\oplus a_1}$} (8.6,-3);
\node[align=center,blue!75!black] at (8.5,3) {$\vec{P}_{fin}^{SM}$};
 \draw[arrows={-stealth}] (6.35,-4.2)--(6.35,-5);  
  \draw[arrows={-stealth}] (7.85,-4.2)--(7.85,-5);  
 \draw[arrows={-stealth}] (8.25,-2)--(8.25,-3);  
\draw[arrows={-stealth},blue!75!black,thick] (6.35,-5)--(6.35,-6); \draw[arrows={-stealth},blue!75!black,thick] (7.85,-5)--(7.85,-6);
\node[align=center,blue!75!black] at (6.35,-6.3) {$\mathbf{a}$}; \node[align=center,blue!75!black] at (7.85,-6.3) {$\mathbf{a'}$}; \draw (1.2,-2) to [out=315,in=180] (6.2,1.8);
\draw (8.25,-2) to [out=90,in=0] (6.2,1.8);
\node[align=center] at (1.8,-2.8) {$a_1$};
\node[align=center] at (0.1,-1.5) {$P_{in}^S$};
\node[align=center] at (7.5,-1.5) {\tiny{$\left(\begin{array}{c}1\\0\\
\hline
1\\
0\\
\end{array}\right)_M$}};
\end{tikzpicture}
    \caption{\textbf{Classical circuit decomposition of memory update transformation $\mathcal{T}$:} The blue box represents the final state of the system $S$ and memory $M$ after the memory update characterised by the fiducial measurements $X$ and $X'$ and the outcomes $a$, $a'$. Let $\mathcal{T}$ be the memory update transformation that maps the initial state $\vec{P}_{in}^{SM}$ to a final state $\vec{P}_{fin}^{SM}$.  Noting that we only need to consider the case of $X=X'$ since the for $X\neq X'$, the entries of $\vec{P}_{fin}^{SM}$ can be arbitrary, the action of $\mathcal{T}$ is eqivalent to the circuit shown here i.e., 1) Choose $X_1=X(=X')$ and perform this fiducial measurement on the initial state of the system $\vec{P}_{in}^S$ to obtain the outcome $a_1$. 2) Fix $X_2 =0$ (or $X_2=1$) and perform this fiducial measurement on the initial state of the memory $\vec{P}_{in}^M=(1\quad 0|1\quad 0)^T_M$ to obtain the outcome $a_2$. 3) Set $a=a_1$. 4) If $a_1=1$, set $a'=a_2$, otherwise set $a'=a_2\oplus 1$, where $\oplus$ denotes modulo 2 addition.}
    \label{fig: memorycircuit}
\end{figure}
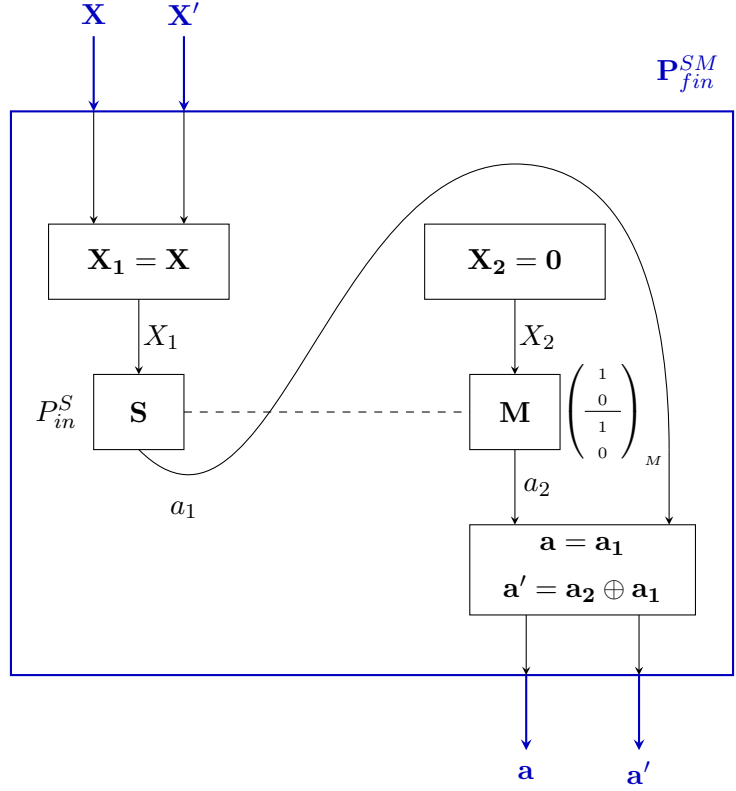

For higher dimensional systems $S$ with $n>2$ fiducial measurements, $X \in \{0,...,n-1\}$ and $k>2$ outcomes taking values $a\in \{0,...,k-1\}$, let $b_n$ and $b_k$ be the number of bits required to represent $n$ and $k$ in binary respectively. Then the memory $M$ would be initialized to $b_k$ copies of the pure state $\vec{P}_{in,n}^M=(1\quad 0|...|1\quad 0)^T_M$ which contains $n$ identical blocks (one for each of the $n$ fiducial measurements). One can then perform the procedure of Figure~\ref{fig: memorycircuit} ''bitwise" combining each output bit with one pure state of $M$ and apply the same argument to obtain the result. For the specific case of the memory update transformation of Figure~\ref{fig:update_copy}, this would correspond to a bitwise CNOT on the output wires of $S$ and $M$.
\end{proof}

\subsection{Two labs sharing initial correlations}

So far, we have considered a single agent measuring a system in her lab. We can also consider situations where multiple agents jointly share a state and measure their local parts of the state, updating their corresponding memories. One might wonder whether the initial correlations in the shared state are preserved once the agents measure it to update their memories (clearly the local measurement probabilities remain unaltered as we saw in this section). The answer is affirmative and this is what allows us to formulate the Frauchiger-Renner paradox in box world as done in the Section~\ref{sec:paradox}, even though a coherent copy analogous to the quantum case does not exist here.

\begin{theorem}
\label{theorem: PRmemory}
Suppose that Alice and Bob share an arbitrary bipartite state $ \vec{P}^{PR}_{in}$ (which may be correlated), locally perform a fiducial measurement on their half of the state and store the outcome in their local memories $A$ and $B$. Then the final joint state $\vec{P}^{\tilde{A}\tilde{B}}_{fin}$ of the systems $\tilde{A}:=PA$ and $\tilde{B}:=RB$ as described by outside agents is isomorphic to  $ \vec{P}^{PR}_{in}$ with the systems $\tilde{A}$ and $\tilde{B}$ taking the role of the systems $P$ and $R$ i.e., local memory updates by Alice and Bob preserve any correlations initially shared between them.
\end{theorem}
\begin{proof}
In the following, we describe the proof for the case where the bipartite system shared by Alice and Bob consists of 2 gbits, however, the result easily generalises to arbitrary higher dimensional systems by the argument presented in the last paragraph of the proof of Theorem~\ref{theorem: boxmemory}. 

Let $\vec{P}^{PR}_{in}$ be an arbitrary 2 gbit state with entries $P_{in}(ab=ij|XY=kl)$ ($i,j,k,l \in \{0,1\}$), which correspond to the joint probabilities of Alice and Bob obtaining the outcomes $a=i$ and $b=j$ when measuring $X=k$ and $Y=l$ on the $P$ and $R$ subsystems when sharing that initial state. Let $X',a'\in \{0,1\}$ and $Y',b' \in \{0,1\}$ be the fiducial measurements and outcomes for the memory systems $A$ and $B$ (also gbits) respectively. We describe the measurement and memory update process for each agent separately and characterise the final state of Alice's and Bob's systems and memories after the process as would appear to outside agents who do not have access to Alice and Bob's measurement outcomes. This analysis does not depend on the order in which Alice and Bob perform the measurement as the correlations are symmetric between them, so without loss of generality, we can consider Bob's measurement first and then Alice's. 

Suppose that Bob's memory $B$ is initialised to the state $\vec{P}_{in}^{B}=\vec{P}_1^{B}=(1\quad 0|1\quad 0)^T_{B}$. Then the joint initial state of the Alice's and Bob's system and Bob's memory as described by an agent Wigner outside Bob's lab is $\vec{P}^{PRB}_{in}=\vec{P}^{PR}_{in}\otimes\vec{P}_1^{B}$. This can be expanded as follows where $P_{in}(abb'=ijk|XYY'=lmn)$ represents the probability of obtaining the binary outcomes $a=i$,$b=j$,$b'=k$ when performing the binary fiducial measurements $X=l$,$Y=m$,$Y'=n$ on the initial state $\vec{P}^{PRB}_{in}$.
\begin{equation}
\label{eq: PRBin}
    \vec{P}^{PRB}_{in}=\left(
\begin{array}{c}
 P_{in}(abb'=000|XYY'=000)\\
 P_{in}(abb'=001|XYY'=000)\\
 P_{in}(abb'=010|XYY'=000)\\
 P_{in}(abb'=011|XYY'=000)\\
 P_{in}(abb'=100|XYY'=000)\\
 P_{in}(abb'=101|XYY'=000)\\
 P_{in}(abb'=110|XYY'=000)\\
 P_{in}(abb'=111|XYY'=000)\\
\hline
\cdot\\
\cdot\\
\cdot\\
\hline
 P_{in}(abb'=000|XYY'=111)\\
 P_{in}(abb'=001|XYY'=111)\\
 P_{in}(abb'=010|XYY'=111)\\
 P_{in}(abb'=011|XYY'=111)\\
 P_{in}(abb'=100|XYY'=111)\\
 P_{in}(abb'=101|XYY'=111)\\
 P_{in}(abb'=110|XYY'=111)\\
 P_{in}(abb'=111|XYY'=111)\\
\end{array}
\right)_{PRB}=\left(
\begin{array}{c}
 P_{in}(ab=00|XY=00)\\
 0\\
 P_{in}(ab=01|XY=00)\\
  0\\
 P_{in}(ab=10|XY=00)\\
  0\\
 P_{in}(ab=11|XY=00)\\
  0\\
\hline
\cdot\\
\cdot\\
\cdot\\
 \hline P_{in}(ab=00|XY=11)\\
  0\\
 P_{in}(ab=01|XY=11)\\
  0\\
 P_{in}(ab=10|XY=11)\\
  0\\
 P_{in}(ab=11|XY=11)\\
\end{array}
\right)_{PRB} 
\end{equation}
$\vec{P}^{PRB}_{in}$ has 8 blocks $G_{XYY'}$, one for each value of $(X,Y,Y')$ and is a product state with 4 equal pairs of blocks, $G^{in}_{000}=G^{in}_{001}$,$G^{in}_{010}=G^{in}_{011}$, $G^{in}_{100}=G^{in}_{101}$, $G^{in}_{110}=G^{in}_{111}$ since both measurements on the initial state of $B$ give the same outcome.

Now, the outside observer Wigner will describe the transformation on $RB$ through 
the memory update transformation $\mathcal{T}$ of Figure~\ref{fig: memorycircuit}. Let $ \vec{P}^{PRB}_{fin}$ be the final state that results by applying this map to the systems $RB$ in the initial state  $\vec{P}^{PRB}_{in}$. Any transformation on a system characterised by $n$ fiducial measurements with $k$ outcomes each can be represented by a $nk\times nk$ block matrix where each block is a $k\times k$ matrix (see \cite{Barrett07} for further details), for the system $RB$, $n=k=4$ and the memory update transformation $\mathcal{T}_{RB}$ would be a $16\times 16$ block matrix of the following form where each $T_{ij}$ is a $4\times 4$ matrix.

\begin{equation*}
\mathcal{T}_{RB}=
\left(
\begin{array}{c|c|c}
T_{11} & \cdot\quad\cdot\quad\cdot & T_{14} \\
\hline
\cdot&\quad &\cdot\\
\cdot&\quad &\cdot\\
\cdot&\quad &\cdot\\
\hline
T_{41} & \cdot\quad\cdot\quad\cdot & T_{44}
\end{array}
\right)_{RB}
\end{equation*}
Here, the first 4 rows decide the entries in the first block of the transformed matrix, the next 4, the second block and so on. Noting that the memory update transformation (Figure~\ref{fig: memorycircuit}) merely permutes elements within the relevant blocks (and does not mix elements between different blocks), the only non-zero blocks of $\mathcal{T}_{RB}$ are the diagonal ones $T_{ii}$. Further, by the same argument as in Theorem~\ref{theorem: boxmemory}, the only relevant entries in the transformed state are when the same fiducial measurement is performed on Bob's system $R$ and memory $B$ i.e., only cases where $Y=Y'$. The remaining measurement choices maybe arbitrary for the final state (just as they are for $X\neq X'$ in Equation~\ref{eq: PSMfin}). This means that among the 4 diagonal blocks, only 2 of them are relevant. The 4 fiducial measurements on $RB$ are $YY'=00,01,10,11$ and in that order, only the first and fourth are relevant since they correspond to $Y=Y'$. Within these relevant blocks (in this case $T_{11}$ and $T_{44}$), the operation is a CNOT on the output $b'$ controlled by the output $b$ and we have the following matrix representation of the memory update map $\mathcal{T}$ of Figure~\ref{fig: memorycircuit}\footnote{The memory update map corresponding to the circuit of Figure~\ref{fig:update_copy} is a specific case of this map where the arbitrary blocks $*$ are also equal to $CN$}.

\begin{equation}
 \mathcal{T}_{RB}=
\left(
\begin{array}{c|c|c|c}
CN & 0 & 0 & 0 \\
\hline
0&*&0&0\\
\hline
0&0&*&0\\
\hline
0 & 0&0 & CN
\end{array}
\right)_{RB},  \qquad CN= \left(
\begin{array}{cccc}
1&0&0&0\\
0&1&0&0\\
0&0&0&1\\
0&0&1&0\\
\end{array}\right)
\end{equation}
where $0$ represents the $4\times 4$ null matrix and blocks labelled $*$ can be arbitrary. The final state $\vec{P}_{fin}^{PRB}$ as seen by Wigner is then
\begin{equation}
   \vec{P}_{fin}^{PRB}= (\mathcal{I}_P\otimes\mathcal{T}_{RB})\vec{P}_{in}^{PRB}=(\mathcal{I}_P\otimes\mathcal{T}_{RB})\Big[\vec{P}_{in}^{PR}\otimes \left(\begin{array}{c}
        1 \\
        0\\
        \hline
        1\\
        0\\
   \end{array}\right)_{B}\Big],
\end{equation}
where $\mathcal{I}_P$ is the identity transformation on the $P$ system. Since the $CN$ blocks are the only relevant blocks in $\mathcal{T}_{RB}$ and each block of $\vec{P}_{in}^{PRB}$ has the same pattern of non-zero and zero entries (Equation~\ref{eq: PRBin}), it is enough to look at the action of $\mathcal{I}_P\otimes CN$ on the first block $G^{in}_{000}$ of $\vec{P}_{in}^{PRB}$. Noting that $\mathcal{I}_P$ is a $2\times 2$ identity matrix, we have

\begin{align*}
\begin{split}
( \mathcal{I}_P\otimes CN)G^{in}_{000}&=\left(\begin{array}{cccccccc}
    1&0&0&0&0&0&0&0 \\
    0&1&0&0&0&0&0&0 \\
    0&0&0&1&0&0&0&0 \\
    0&0&1&0&0&0&0&0 \\
    0&0&0&0&1&0&0&0 \\
    0&0&0&0&0&1&0&0 \\
    0&0&0&0&0&0&0&1 \\
    0&0&0&0&0&0&1&0 \\
\end{array}\right) \left(
\begin{array}{c}
 P_{in}(ab=00|XY=00)\\
 0\\
 P_{in}(ab=01|XY=00)\\
  0\\
 P_{in}(ab=10|XY=00)\\
  0\\
 P_{in}(ab=11|XY=00)\\
  0\\
\end{array}
\right)=G^{fin}_{000}\\
&=\left(
\begin{array}{c}
 P_{in}(ab=00|XY=00)\\
 0\\
 0\\
 P_{in}(ab=01|XY=00)\\
 P_{in}(ab=10|XY=00)\\
  0\\
  0\\
 P_{in}(ab=11|XY=00)\\
\end{array}
\right)
=\left(
\begin{array}{c}
 P_{fin}(abb'=000|XYY'=000)\\
 P_{fin}(abb'=001|XYY'=000)\\
 P_{fin}(abb'=010|XYY'=000)\\
 P_{fin}(abb'=011|XYY'=000)\\
 P_{fin}(abb'=100|XYY'=000)\\
 P_{fin}(abb'=101|XYY'=000)\\
 P_{fin}(abb'=110|XYY'=000)\\
 P_{fin}(abb'=111|XYY'=000)\\
\end{array}
\right),
\end{split}
\end{align*}

where $P_{fin}(abb'=ijk|XYY'=lmn)$ represents the probability of obtaining the outcomes $a=i$,$b=j$,$b'=k$ when performing the fiducial measurements $X=l$,$Y=m$,$Y'=n$ on the final state $\vec{P}^{PRB}_{fin}$ and $G^{fin}_{000}$ is the first block of this final state. Clearly the only non-zero outcome probabilities are when $b=b'$ and this allows us to compress the final state by defining $\tilde{b}=i\Leftrightarrow b=b'=i$ for $i\in \{0,1\}$ and we have the following.
\begin{equation*}
    ( \mathcal{I}_P\otimes CN)G^{in}_{000} \equiv \left(
\begin{array}{c}
 P_{in}(ab=00|XY=00)\\
 P_{in}(ab=01|XY=00)\\
 P_{in}(ab=10|XY=00)\\
 P_{in}(ab=11|XY=00)\\
\end{array}
\right)=\left(
\begin{array}{c}
 P_{fin}(a\tilde{b}=00|XYY'=000)\\
 P_{fin}(a\tilde{b}=01|XYY'=000)\\
 P_{fin}(a\tilde{b}=10|XYY'=000)\\
 P_{fin}(a\tilde{b}=11|XYY'=000)
\end{array}
\right)=G_{00}^{in}
\end{equation*}
Here $G_{00}^{in}$ is the first block of the initial state $\vec{P}_{in}^{PR}$ and we have that the first block of the final state of $PRB$ is equivalent (up to zero entries) to the first block of the initial state over $PR$ alone or $G^{fin}_{000}=G_{00}^{in}$. Among the 8 blocks of $\vec{P}_{fin}^{PRB}$, only the 4 blocks $G^{fin}_{000}$,$G^{fin}_{011}$,$G^{fin}_{100}$ and $G^{fin}_{111}$ are the relevant ones (since $Y=Y'$ for these) and we can similarly show that $G^{fin}_{011}\equiv G^{in}_{01}$,$G^{fin}_{100}\equiv G^{in}_{10}$ and $G^{fin}_{111}\equiv G^{in}_{11}$ for the remaining 3 relevant blocks. Defining $\tilde{Y}=i \Leftrightarrow Y=Y'=i$ for $i\in\{0,1\}$, we obtain

\begin{equation}
\label{eq: PRBfin}
    \vec{P}_{fin}^{PRB}=\vec{P}_{fin}^{P\tilde{B}}\equiv\left(
\begin{array}{c}
 P_{fin}(a\tilde{b}=00|X\tilde{Y}=00)\\
 P_{fin}(a\tilde{b}=01|X\tilde{Y}=00)\\
 P_{fin}(a\tilde{b}=10|X\tilde{Y}=00)\\
 P_{fin}(a\tilde{b}=11|X\tilde{Y}=00)\\
 \hline
  P_{fin}(a\tilde{b}=00|X\tilde{Y}=01)\\
 P_{fin}(a\tilde{b}=01|X\tilde{Y}=01)\\
 P_{fin}(a\tilde{b}=10|X\tilde{Y}=01)\\
 P_{fin}(a\tilde{b}=11|X\tilde{Y}=01)\\
 \hline
  P_{fin}(a\tilde{b}=00|X\tilde{Y}=10)\\
 P_{fin}(a\tilde{b}=01|X\tilde{Y}=10)\\
 P_{fin}(a\tilde{b}=10|X\tilde{Y}=10)\\
 P_{fin}(a\tilde{b}=11|X\tilde{Y}=10)\\
 \hline
  P_{fin}(a\tilde{b}=00|X\tilde{Y}=11)\\
 P_{fin}(a\tilde{b}=01|X\tilde{Y}=11)\\
 P_{fin}(a\tilde{b}=10|X\tilde{Y}=11)\\
 P_{fin}(a\tilde{b}=11|X\tilde{Y}=11)\\
\end{array}
\right)= \left(
\begin{array}{c}
 P_{in}(ab=00|XY=00)\\
 P_{in}(ab=01|XY=00)\\
 P_{in}(ab=10|XY=00)\\
 P_{in}(ab=11|XY=00)\\
 \hline
  P_{in}(ab=00|XY=01)\\
 P_{in}(ab=01|XY=01)\\
 P_{in}(ab=10|XY=01)\\
 P_{in}(ab=11|XY=01)\\
 \hline
  P_{in}(ab=00|XY=10)\\
 P_{in}(ab=01|XY=10)\\
 P_{in}(ab=10|XY=10)\\
 P_{in}(ab=11|XY=10)\\
 \hline
  P_{in}(ab=00|XY=11)\\
 P_{in}(ab=01|XY=11)\\
 P_{in}(ab=10|XY=11)\\
 P_{in}(ab=11|XY=11)\\
\end{array}
\right)=\vec{P}_{in}^{PR}
\end{equation}
Equation~\ref{eq: PRBfin} shows that final state $\vec{P}_{fin}^{P\tilde{B}}$ of Alice's system $P$, Bob's system $R$ and Bob's memory $B$ after Bob's local memory update is isomorphic to the initial state $\vec{P}_{in}^{PR}$ shared by Alice and Bob, having the same outcome probabilities as the latter for all the relevant measurements. Thus the initial correlations present in $\vec{P}_{in}^{PR}$ are preserved after Bob locally updates his memory according to the update procedure of Figure~\ref{fig: memorycircuit}. One can now repeat the same argument for Alice's local memory update taking $\vec{P}_{fin}^{P\tilde{B}}\otimes (1\quad 0|1 \quad 0)^T_{A}$ to be the initial state and by analogously defining $\tilde{s}=i \Leftrightarrow s=s'=i$ for $s\in \{a,X\}$,$i\in \{0,1\}$, we have the required result that the final state after both parties perform their local memory updates (as described by outside agents Ursula and Wigner) is isomorphic and operationally equivalent to the initial state shared by the parties before the memory update.
\begin{equation}
    \vec{P}_{fin}^{PARB}=\vec{P}_{fin}^{\tilde{A}\tilde{B}}\equiv \vec{P}_{in}^{PR}
\end{equation}

\end{proof}

\section{Quantum measurements in GPT language}
\label{appendix:fr}
In the PR box analysis, we encounter a peculiarity which is specific to measurement procedures in GPTs: the box ``disappears'' after it
is measured. This can become a problem when, during the course of the experiment, the observer measuring the box has to be measured together with the box. This is the case in the original Frauchiger-Renner thought experiment.
However, this issue can in principle be avoided, if one adapts the description of the experiment to the mentioned peculiarity: as soon as the agent measures the box, and it subsequently disappears, she prepares a new box for the observer on the outside to measure. For example, when Alice measures the box $P$, she can not only prepare a box $R_a$ for Bob to measure (Figure \ref{fig:alice-bob}), but also one for Wigner, meant to contain correlations of the Bob's lab (Figure \ref{fig:alice-wigner}). Similarly, from Bob's point of view, he prepares a box $PA_b$ for Ursula to measure (Figure \ref{fig:bob-ursula}); and, finally, as seen from the outside, Ursula and Wigner measure boxes $PA_b$ and $RB_a$, prepared for them by Bob and Alice (Figure \ref{fig:outside}).

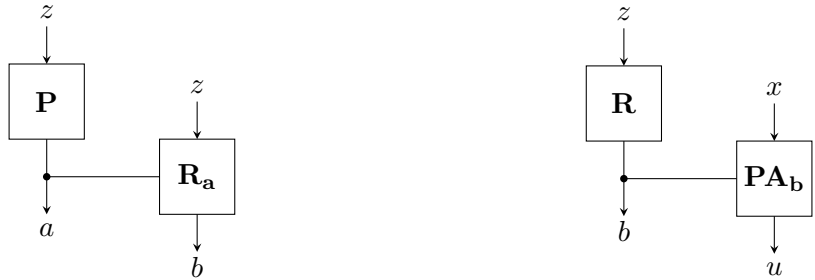
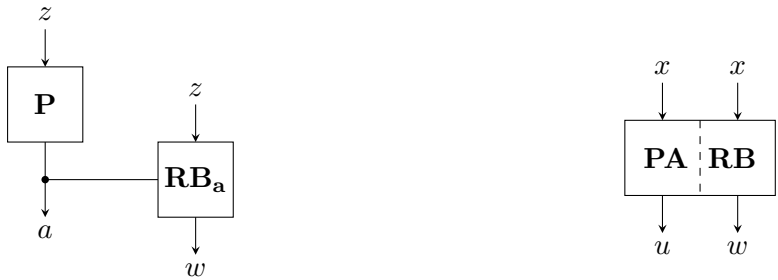
\begin{figure}[h]
\centering
    \begin{subfigure}{0.45\textwidth}
        \begin{align*}
        \begin{tikzpicture}
        \draw (0,0) rectangle node[align=center]{$\mathbf{P}$} (1,1);
        \draw[arrows={-stealth}] (0.5,1.5)--(0.5,1);
        \node[align=center] at (0.5,1.7) {$z$};
        \draw[arrows={-stealth}]  (0.5,0)--(0.5,-1);
        \node[align=center] at (0.5, -1.2) {$a$};
        \node[circle,fill,inner sep=1pt] at (0.5,-0.5) {};
        \draw (0.5,-0.5) -- (2,-0.5);
        \draw (2,0) rectangle node[align=center]{$\mathbf{R_a}$} (3,-1);
        \draw[arrows={-stealth}] (2.5,0.5)--(2.5,0);
        \node[align=center] at (2.5,0.7) {$z$};
        \draw[arrows={-stealth}] (2.5,-1)--(2.5,-1.5);
        \node[align=center] at (2.5, -1.7) {$b$};
        %\draw [decorate, decoration={brace,amplitude=30pt,raise=1pt,mirror}] (2.5, -1.7) -- (3.5, -1.7);
        \end{tikzpicture}
        \end{align*}
        \caption{Alice's viewpoint: Alice measures the box $P$ and prepares a box $R_a$ for Bob to measure.}
        \label{fig:alice-bob}
    \end{subfigure}
    \
    \begin{subfigure}{0.45\textwidth}
        \begin{align*}
        \begin{tikzpicture}
        \draw (0,0) rectangle node[align=center]{$\mathbf{R}$} (1,1);
        \draw[arrows={-stealth}] (0.5,1.5)--(0.5,1);
        \node[align=center] at (0.5,1.7) {$z$};
        \draw[arrows={-stealth}]  (0.5,0)--(0.5,-1);
        \node[align=center] at (0.5, -1.2) {$b$};
        \node[circle,fill,inner sep=1pt] at (0.5,-0.5) {};
        \draw (0.5,-0.5) -- (2,-0.5);
        \draw (2,0) rectangle node[align=center]{$\mathbf{PA_b}$} (3,-1);
        \draw[arrows={-stealth}] (2.5,0.5)--(2.5,0);
        \node[align=center] at (2.5,0.7) {$x$};
        \draw[arrows={-stealth}] (2.5,-1)--(2.5,-1.5);
        \node[align=center] at (2.5, -1.7) {$u$};
        %\draw [decorate, decoration={brace,amplitude=30pt,raise=1pt,mirror}] (2.5, -1.7) -- (3.5, -1.7);
        \end{tikzpicture}
        \end{align*}
        \caption{Bob's viewpoint: Bob measures the box $R$ and prepares a box $PA_b$ for Ursula to measure.}
        \label{fig:bob-ursula}
    \end{subfigure}
    \\
    \begin{subfigure}{0.45\textwidth}
        \begin{align*}
        \begin{tikzpicture}
        \draw (0,0) rectangle node[align=center]{$\mathbf{P}$} (1,1);
        \draw[arrows={-stealth}] (0.5,1.5)--(0.5,1);
        \node[align=center] at (0.5,1.7) {$z$};
        \draw[arrows={-stealth}]  (0.5,0)--(0.5,-1);
        \node[align=center] at (0.5, -1.2) {$a$};
        \node[circle,fill,inner sep=1pt] at (0.5,-0.5) {};
        \draw (0.5,-0.5) -- (2,-0.5);
        \draw (2,0) rectangle node[align=center]{$\mathbf{RB_a}$} (3,-1);
        \draw[arrows={-stealth}] (2.5,0.5)--(2.5,0);
        \node[align=center] at (2.5,0.7) {$z$};
        \draw[arrows={-stealth}] (2.5,-1)--(2.5,-1.5);
        \node[align=center] at (2.5, -1.7) {$w$};
        %\draw [decorate, decoration={brace,amplitude=30pt,raise=1pt,mirror}] (2.5, -1.7) -- (3.5, -1.7);
        \end{tikzpicture}
        \end{align*}
        \caption{Alice's viewpoint: after measuring the box $P$, she also prepares a box $RB_a$ for Wigner to measure.}
        \label{fig:alice-wigner}
    \end{subfigure}
    \
    \begin{subfigure}{0.45\textwidth}
        \begin{align*}
        \begin{tikzpicture}
        \draw (0,0) rectangle node[align=center]{$\mathbf{PA}$ \ $\mathbf{RB}$} (2,1);
        \draw[arrows={-stealth}] (0.5,1.5)--(0.5,1);
        \node[align=center] at (0.5,1.7) {$x$};
        \draw[arrows={-stealth}] (1.5,1.5)--(1.5,1);
        \node[align=center] at (1.5,1.7) {$x$};
        \draw [dashed] (1,1)--(1,0);
        \draw[arrows={-stealth}] (0.5,0)--(0.5,-0.5);
        \node[align=center] at (0.5,-0.7) {$u$};
        \draw[arrows={-stealth}] (1.5,0)--(1.5,-0.5);
        \node[align=center] at (1.5,-0.7) {$w$};
        \end{tikzpicture}
        \end{align*}
        \caption{Ursula's and Wigner's viewpoints: they measure boxes $PA$ and $RB$ respectively.}
        \label{fig:outside}
    \end{subfigure}
\caption{\textbf{Viewpoints of different agents for quantum measurements in GPTs.} }
\label{fig:}
\end{figure}

\end{document}